\documentclass{new_tlp}

\DeclareMathAlphabet{\mathcal}{OMS}{cmsy}{m}{n}
\newcommand\bcmdtab{\noindent\bgroup\tabcolsep=0pt%
  \begin{tabular}{@{}p{10pc}@{}p{20pc}@{}}}
\newcommand\ecmdtab{\end{tabular}\egroup}

\usepackage{bm}
\usepackage{mathptmx}
\usepackage{multirow}
\usepackage{xspace}
\usepackage{amsmath}
\usepackage{graphicx}
\usepackage{amssymb}
\usepackage[usenames, dvipsnames]{color}
\usepackage{subfigure}
\usepackage{xspace}
\usepackage{comment}
\usepackage{url}
\usepackage{bigstrut}
\usepackage{tikz}

\newcommand{\rp}{\mathfrak{R}}

\newcommand{\shy}{\mathsf{shy}}

\newcommand{\terms}{\mathit{terms}}
\newcommand{\const}{\mathit{const}}
\newcommand{\mods}{\mathit{mods}}
\newcommand{\fmods}{\mathit{fmods}}

\newcommand{\EV}{\mathbf{EV}}
\newcommand{\UV}{\mathbf{UV}}
\newcommand{\V}{\mathbf{V}}
\newcommand{\C}{\mathbf{C}}

\newcommand{\Y}{\mathbf{Y}}
\newcommand{\smooth}{smooth}
\newcommand{\CS}{\mathcal{CS}}

\newcommand{\new}[1]{#1}
\newcommand{\rem}[1]{}

\newcommand{\beforeProof}{\vspace{-0em}}
\newcommand{\beforeSec}{\vspace{-0em}}
\newcommand{\beforeSubsec}{\vspace{-0em}}
\newcommand{\befAftVectors}{\vspace{-0em}}
\newcommand{\beforeCaption}{\vspace{-0em}}
\newcommand{\afterCaption}{\vspace{-0em}}

\title[Finite model reasoning over existential rules]
  {Finite model reasoning over existential rules}

  \author[G. Amendola, N. Leone, M. Manna]
         {GIOVANNI AMENDOLA, NICOLA LEONE, MARCO MANNA\\
         Department of Mathematics and Computer Science, University of Calabria, Italy\\
         \email{$\{$amendola,leone,manna$\}$@mat.unical.it}}

\jdate{March 2003}
\pubyear{2003}
\pagerange{\pageref{firstpage}--\pageref{lastpage}}
\doi{S1471068401001193}

\newtheorem{lemma}{Lemma}[section]
\newtheorem{definition}{Definition}[section]
\newtheorem{theorem}{Theorem}[section]

\newtheorem{example}{Example}[section]

\newtheorem{proposition}{Proposition}[section]

\usepackage{color}
\usepackage{todonotes}

\begin{document}

\label{firstpage}

\maketitle

\begin{abstract}
Ontology-based query answering (OBQA) asks whether a Boolean conjunctive query is satisfied by all models of a logical theory consisting of a relational database paired with an ontology. The introduction of existential rules (i.e., Datalog rules extended with existential quantifiers in rule-heads) as a means to specify the ontology gave birth to \mbox{Datalog+/-,} a framework that has received increasing attention in the last decade, with focus also on decidability and finite controllability to support effective reasoning. Five basic decidable fragments have been singled out: linear, weakly-acyclic, guarded, sticky, and shy. Moreover, for all these fragments, except shy, the important property of finite controllability has been proved, ensuring that a query is satisfied by all models of the theory iff it is satisfied by all its finite models. In this paper we complete the picture by demonstrating that finite controllability of OBQA holds also for shy ontologies, and it therefore applies to all basic decidable Datalog+/- classes. To make the demonstration, we devise a general technique to facilitate the process of (dis)proving finite controllability of an arbitrary ontological fragment.
\end{abstract}



\newcommand{\nop}[1]{}
\newcommand{\fmodels}{\,{\models_{\mathsf{fin}}}\,}
\newcommand{\wsfmodels}{\,{\models_{\mathsf{wsf}}}\,}

\beforeSec

\section{Introduction}
\label{sec:intro}

The problem of answering a 
Boolean query $q$ against a logical theory consisting of an extensional database $D$
paired with an ontology $\Sigma$ 
is attracting the increasing attention of scientists in various fields 
of Computer Science, ranging from 
Artificial Intelligence~\cite{DBLP:journals/ai/BagetLMS11,%
	DBLP:journals/ai/CalvaneseGLLR13,%
	DBLP:journals/ai/GottlobKKPSZ14} 
to Database Theory~\cite{DBLP:journals/tods/BienvenuCLW14,%
	DBLP:journals/tods/GottlobOP14,%
	Bourhis:2016:GDT:3014437.2976736} 
and Logic~\cite{PerezUrbina2010186,%
	DBLP:journals/corr/BaranyGO13,%
	DBLP:conf/icalp/GottlobPT13ICALP}. 
This problem, called {\em ontology-based query answering}, for short
OBQA~\cite{DBLP:conf/dlog/CaliGL09},
is usually stated as \mbox{$D \cup \Sigma \models q$,}
and it is equivalent to checking whether $q$ is satisfied 
by all models of $D \cup \Sigma$ according to the standard approach of first-order logics, yielding an open world semantics.

Description Logics~\cite{Baader:2003:DLH:885746} and Datalog$^\pm$~\cite{DBLP:conf/icdt/CaliGL09}
have been recognized as the two main families 
of formal knowledge representation languages
to specify $\Sigma$, while
union of (Boolean) conjunctive queries, U(B)CQs for short,
is the most common and studied formalism to express $q$.
For both these families, OBQA is
generally undecidable \cite{DBLP:conf/icdt/Rosati07,DBLP:journals/jair/CaliGK13}.
Hence, a number of syntactic decidable fragments 
of the above ontological languages have been singled out.
However, decidability alone is not the only desideratum.
%
For example, a good balance between computational complexity and
expressive power is, without any doubt, of high importance too.
But there is another property that is turning out to be as 
interesting as challenging to prove:
it goes under the name of {\em finite controllability}~\cite{Rosati:2006:DFC:1142351.1142404}.
An ontological fragment $\mathcal{F}$ is finitely controllable if,
for each triple $\langle D, \Sigma, q\rangle$ with $\Sigma \in \mathcal{F}$, it holds that
$D \cup \Sigma \not\models q$ implies that there is a finite
model $M$ of $D \cup \Sigma$ such that $M \not\models q$.
This is usually stated as $D \cup \Sigma \models q$ if, and 
only if, $D \cup \Sigma \fmodels q$ 
(where $\fmodels$ stands for entailment under finite models),
as the ``only if'' direction is always trivially true.
%
And there are contexts, like in  databases~\cite{DBLP:journals/jcss/JohnsonK84,Rosati:2006:DFC:1142351.1142404,DBLP:journals/corr/BaranyGO13}, 
in which reasoning with respect to finite models is preferred.

\newcommand{\mysize}[1]{{\small #1}}
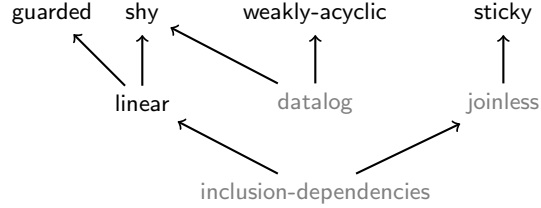
\begin{figure}[t!]
	\begin{center}
		\begin{tikzpicture}
		[->,shorten >=1pt,auto,thick,main node/.style={font=\sffamily\bfseries}]
		
		\node[main node] (1) {\mysize{$\mathsf{shy}$}};
		\node[main node] (6) [right of=1, node distance=2.3cm] {$\mathsf{weakly}$-$\mathsf{acyclic}$};		
		\node[main node] (2) [right of=6,node distance=2.5cm] {\mysize{$\mathsf{sticky}$}};
		\node[main node] (4) [below of=6,node distance=1.2cm] {\textcolor{gray}{\mysize{$\mathsf{datalog}$}}};
		\node[main node] (3) [left of=4,node distance=2.3cm] {\mysize{$\mathsf{linear}$}};		
		\node[main node] (5) [left of=1,node distance=1.2cm] {\mysize{$\mathsf{guarded}$}};
		\node[main node] (7) [right of=4,node distance=2.5cm] {\textcolor{gray}{\mysize{$\mathsf{joinless}$}}};
		\node[main node] (8) [below of=4,node distance=1.2cm] {\textcolor{gray}{\mysize{$\mathsf{inclusion}$-$\mathsf{dependencies}$}}};			
		\draw (4) to node {} (1);
		\draw (3) to node {} (1);
		\draw (3) to node {} (5);
		\draw (4) to node {} (6);
		\draw (7) to node {} (2);
		\draw (8) to node {} (3);
		\draw (8) to node {} (7);
		\end{tikzpicture}
	\end{center}
	\beforeCaption
	\caption{Taxonomy of the basic Datalog$^\pm$ classes.}\label{fig:taxonomy}
	\afterCaption
\end{figure}

In this paper we focus on the Datalog$^\pm$ family, 
which has been introduced with the aim of ``closing the gap between the Semantic
Web and databases''~\cite{DBLP:journals/ws/CaliGL12} to provide the {\em Web of Data}
with scalable formalisms that can benefit from existing database technologies.
In fact, Datalog$^\pm$ generalizes two well-known subfamilies of Description Logics called $\mathcal{EL}$ and {\em DL-Lite},
which collect the basic tractable languages for OBQA
in the context of the Semantic Web and databases.
In particular, we consider ontologies where 
$\Sigma$ is a set of {\em existential rules}, 
each of which is a first-order formula $\rho$ of the form
$
\forall {\bf X} \forall {\bf Y} ( \phi({\bf X},{\bf Y}) \rightarrow \exists {\bf Z} p({\bf X},{\bf Z}))$,
where the {\em body} $\phi({\bf X},{\bf Y})$ of $\rho$ is a conjunction of atoms, and 
the {\em head} $p({\bf X},{\bf Z})$ of $\rho$ is a single atom.

The main decidable Datalog$^\pm$ fragments rely on the following five syntactic properties:
{\em weak-acyclicity}~\cite{DBLP:journals/tcs/FaginKMP05}, 
{\em guardedness}~\cite{DBLP:journals/jair/CaliGK13}, 
{\em linearity}~\cite{DBLP:journals/ws/CaliGL12}, 
{\em stickiness}~\cite{DBLP:journals/pvldb/CaliGP10}, %
and {\em shyness}~\cite{DBLP:conf/kr/LeoneMTV12}. 
And these properties underlie
the basic classes called
$\mathsf{weakly}$-$\mathsf{acyclic}$,
$\mathsf{guarded}$,
$\mathsf{linear}$,
$\mathsf{sticky}$, and 
$\mathsf{shy}$, respectively.
Several variants and combinations of these classes
have been defined and studied too~\cite{DBLP:conf/kr/BagetLM10,%
	DBLP:conf/ijcai/KrotzschR11,%
	DBLP:journals/ai/CaliGP12,%
	DBLP:conf/datalog/CiviliR12,%
	DBLP:journals/tplp/GottlobMP13},
as well as semantic properties subsuming the syntactic ones 
\cite{DBLP:conf/ijcai/BagetLMS09,DBLP:conf/kr/LeoneMTV12}.

The five basic classes above are pairwise uncomparable, except for
$\mathsf{linear}$ which is strictly contained in both $\mathsf{guarded}$ and $\mathsf{shy}$, as depicted in Figure~\ref{fig:taxonomy}.
Interestingly, both $\mathsf{weakly}$-$\mathsf{acyclic}$ and $\mathsf{shy}$ strictly contain
$\mathsf{datalog}$ ---the well-known class with rules 
of the form $\forall {\bf X} \forall {\bf Y} ( \phi({\bf X},{\bf Y}) \rightarrow p({\bf X}))$,
where existential quantification has been dropped.
Moreover, $\mathsf{sticky}$ strictly contains $\mathsf{joinless}$ ---the class collecting sets of rules where each body contains no repeated variable.
The latter, introduced by \citeN{DBLP:conf/lics/GogaczM13}
to prove that $\mathsf{sticky}$ is finitely controllable, plays a
central role also in this paper.
Finally, both $\mathsf{linear}$ and $\mathsf{joinless}$ strictly contain
$\mathsf{inclusion}$-$\mathsf{dependencies}$ ---the well-known class of 
relational database dependencies collecting sets of rules 
with one single body atom and no repeated variable.

Under arbitrary models, OBQA can be
reduced to the problem of answering $q$ over a universal (or canonical)
model $U$ that can be homomorphically embedded into every other model (both 
finite and infinite) of $D \cup \Sigma$.
Therefore, $D \cup \Sigma \models q$ if, and only if,  $U \models q$.
A way to compute a universal model is to employ the so called {\em chase} procedure.
%
Starting from $D$, the chase
``repairs'' violations of rules
by repeatedly adding new atoms ---introducing fresh values, called {\em nulls}, whenever required by an existential variable--- 
until a fixed-point satisfying all rules is reached.
In the classical setting, the chase is therefore sound and complete.
But when {\em finite model reasoning} \new{(namely reasoning over finite models only, here denoted by} $\fmodels$) is required, then the chase 
is generally uncomplete, unless ontologies are finitely controllable.
%
%
Hence, proving this property is of utmost importance, 
especially \new{in} those contexts where finite model reasoning is relevant.

%
%
%
%

Finite controllability of $\mathsf{weakly}$-$\mathsf{acyclic}$ 
comes for free since every ontology here admits a finite universal model, 
computed by a variant of the chase procedure which goes under the name of 
{\em restricted chase}~\cite{DBLP:journals/tcs/FaginKMP05}.
Conversely, the proof of this property for the subsequent three classes 
has been a very different matter.
Complex, yet intriguing, constructions
have been devised for
$\mathsf{linear}$~\cite{Rosati:2006:DFC:1142351.1142404,DBLP:journals/corr/BaranyGO13},
$\mathsf{guarded}$~\cite{DBLP:journals/corr/BaranyGO13}, and more recently for
$\mathsf{sticky}$~\cite{DBLP:conf/lics/GogaczM13}.
To complete the picture, we have addressed the same problem for $\mathsf{shy}$
and get the following positive result, which is the main contribution of the paper.

\begin{theorem}\label{th:main}	
	Under $\mathsf{shy}$ ontologies, $D \cup\Sigma \models q$ if, and only if,
	$D\cup\Sigma\fmodels q$.
\end{theorem}

%
For the proof, we design in Section~\ref{sec:tools}
and exploit in Section~\ref{sec:challenge}
a general technique (our second contribution), called {\em canonical rewriting},
to facilitate the process of (dis)proving finite controllability
of an arbitrary ontological fragment of existential rules.
By exploiting this technique, we can 
immediately (re)confirm that $\mathsf{linear}$ is finitely controllable since
$\mathsf{inclusion}$-$\mathsf{dependencies}$ is.
In addition, we prove (our third contribution) that
$\mathsf{sticky}$-$\mathsf{join}$\new{~\cite{DBLP:journals/ai/CaliGP12}, generalizing both $\mathsf{sticky}$ and $\mathsf{linear}$,}
 is finitely controllable since
$\mathsf{sticky}$ is.
However, differently from $\mathsf{linear}$ and 
$\mathsf{sticky}$-$\mathsf{join}$, 
the canonical rewriting of a $\mathsf{shy}$ ontology ---although it is \new{simpler and} still a $\mathsf{shy}$ ontology--- does not \new{immediately} fall in any other known class.
Therefore, to prove that $\shy$ is finitely controllable,
we devise three technical tools on top of the canonical rewriting
from which we are able to exploit the fact that $\mathsf{joinless}$
is finitely controllable.

%
%
%
%
%

\newcommand{\myPar}[1]{{\bf \em #1.}}

\beforeSec
\section{Ontology-based query answering}
\label{sec:Prel}

\myPar{Basics}~Let $\bf{C}$, $\bf{N}$ and $\bf{V}$ 
denote pairwise disjoint discrete sets
of \textit{constants}, \textit{nulls} and \textit{variables},
respectively.
An element $t$ of ${\bf T} = {\bf C\cup N\cup V}$ is called \textit{term}. 
%
%
An \textit{atom} \new{$\alpha$} is a labeled tuple $p(t_1,\ldots,t_m)$, where 
$p$ is a predicate symbol,
$m$ is the {\em arity} \new{of both $p$ anf $\alpha$}, and 
$t_1,\ldots,t_m$ are terms.
An atom is \textit{simple} if it contains no repeated term.
\new{We denote by $pred(\alpha)$ the predicate symbol $p$, and by $\alpha[i]$ the $i$-th term $t_i$ of the $\alpha$.}
We also consider {\em propositional} atoms, which are simple atoms of arity $0$ 
written without brackets. 
Given two sets $A$ and $B$ of atoms, a homomorphism from $A$ to $B$ 
is a mapping $h:{\bf T}\rightarrow{\bf T}$ such that
$c\in {\bf C}$ implies $h(c)=c$,
and also $p(t_1,\ldots,t_m)\in A$ implies $p(h(t_1),\ldots,h(t_m))\in B$.
As usual\new{, we denote by} $h(A) = \{ p(h(t_1),\ldots,h(t_m)): p(t_1,\ldots,t_m)\in A\}\subseteq B$.
%
%
An \textit{instance} $I$ is a discrete set of atoms where each term is either a constant or a null.

\myPar{Syntax}~A \textit{database} $D$ is a finite null-free instance.
An {\em (existential)} {\em rule} $\rho$
is a first-order formula  
$
\forall {\bf X} \forall {\bf Y} ( \phi({\bf X},{\bf Y}) \rightarrow \exists {\bf Z} p({\bf X},{\bf Z}))$,
%
where $body(\rho) = \phi({\bf X},{\bf Y})$ is a conjunction of atoms,
and $head(\rho) = p({\bf X},{\bf Z})$ is an atom. 
Constants may occur in $\rho$.
If ${\bf Z} = \emptyset$, then $\rho$ is $\mathsf{datalog}$ rule. 
An {\em ontology} $\Sigma$ is a set  of rules.
\new{For each rule $\rho$ of $\Sigma$, we denote by ${\bf V}(\rho)$ the set of variables appearing in $\rho$, by $\EV(\rho)$ the set of all existential variables of $\rho$, and
by $\UV(\rho)$ the set of all universal variables of $\rho$.}
A \textit{union of Boolean conjunctive query}, UBCQ for short, $q$ is a first-order expression 
of the form
$\exists {\bf Y}_1 \psi_1({\bf Y}_1) \vee \ldots \vee \exists {\bf Y}_k \psi_k({\bf Y}_k)$,
where 
each $\psi_j({\bf Y}_j)$ is a conjunction of atoms.
%
Constants may occur also in $q$. 
In case $k = 1$, then $q$ is simply called BCQ.
%

%

%
\myPar{Semantics}~Consider a triple $\langle D, \Sigma, q \rangle$ as above.
An instance $I$ {\em satisfies} a rule $\rho \in \Sigma$, denoted by $I\models \rho$, if whenever there is a homomorphism $h$ 
from $body(\rho)$ to $I$,
then there is a homomorphism $h'\supseteq h|_{\bf X}$ 
from $\{head(\rho)\}$ to $I$.
Moreover, $I$ {\em satisfies} $\Sigma$, 
denoted by $I\models \Sigma$, if $I$ satisfies each rule of $\Sigma$.
The  {\em models} of $D\cup\Sigma$, denoted by $\mathit{mods}(D,\Sigma)$,
consist of the set $\{I: I\supseteq D\mbox{ and } I\models \Sigma\}$.
An instance $I$ {\em satisfies} $q$, written $I \models q$, if
there is a homomorphism from some $\psi_j({\bf Y}_j)$ to $I$.
%
Also, $q$ is {\em true} over $D \cup \Sigma$, written $D\cup \Sigma\models q$,
if each model of $D\cup \Sigma$ satisfies $q$.

\myPar{The chase}~Consider a logical theory $\langle D, \Sigma\rangle$ as above.
\new{A rule $\rho$ of $\Sigma$ is {\em applicable} to an instance $I$ if there is
	a homomorphism $h$ from $body(\rho)$ to $I$
	that maps the existential variables of $\rho$
	to different nulls not occurring in $I$.
	If so, $\langle \rho, h\rangle (I) = I \cup h(head(\rho))$
	defines a {\em chase step}.
The {\em chase procedure}~\cite{DBLP:conf/pods/DeutschNR08}
of $D \cup \Sigma$ is any sequence 
$I_0=D\subset I_1\subset\ldots\subset I_m\subset\ldots$ of instances 
obtained by applying exhaustively the rules of $\Sigma$ in a fair \new{(e.g., breadth-first)} fashion in such a way that, for each $i>0$, $\langle \rho,h\rangle (I_{i-1}) = I_i$
defines a chase step for some $\rho$ and $h$.
We call $chase(D,\Sigma)$ the (possibly infinite) instance $\bigcup_{i\geqslant 0} I_i$.}
%
%
\new{Importantly, different chase steps introduce different nulls. 
This variant of the chase is called {\em oblivious},
and defines a family of isomorphic instances, namely
any two such instances are equal modulo renaming of nulls.
Hence, without loss of generality, it is common practice to consider 
the oblivious chase as deterministic and its least fixpoint
as unique. 
The {\em restricted} version of this procedure imposes a further condition 
on each chase step: $I \not \models h'(head(\rho))$,
where $h' = h|_{\UV(\rho)}$. 
Differently from the oblivious one, it defines a family of homomorphically equivalent instances, each generically denoted by $rchase(D,\Sigma)$.
It is well-known that ($r$)$chase(D,\Sigma)$ is a \textit{universal} model of $D \cup \Sigma$, namely for each $M\in \mathit{mods}(D,\Sigma)$, there is a homomorphism from $chase(D,\Sigma)$ to $M$.
Hence, given a UBCQ $q$, it holds that
($r$)$chase(D,\Sigma)\models q$ if, and only if, $D\cup\Sigma\models q$ \mbox{\cite{DBLP:journals/tcs/FaginKMP05}.}}


\rem{
\myPar{Shy ontologies}~Consider a chase step \mbox{$\langle \rho, h\rangle (I)=I'$} employed in the construction of $\mathit{chase}(D,\Sigma)$.
If $\Sigma$ is $\shy$, then 
its underlying properties~\cite{DBLP:conf/kr/LeoneMTV12} guarantee that:
$(1)$ if $X$ occurs in two different atoms of $body(\rho)$, then $h(X) \in {\bf C}$; and
$(2)$ if $X$ and $Y$ occur both in $head(\rho)$ and in two different atoms of $body(\rho)$, then
\mbox{$h(X) = h(Y)$ implies $h(X) \in {\bf C}$.}}

\myPar{Finite controllability}~The {\em finite models} of a theory $D\cup\Sigma$, denoted by $\mathit{fmods}(D,\Sigma)$, are the finite instances in
$\{I\in \mathit{mods}(D,\Sigma): |I|\in \mathbb{N}\}$.
An ontological fragment $\mathcal{F}$ is {\em finitely controllable} if,
for each database $D$,
for each ontology $\Sigma$ of $\mathcal{F}$, and
for each UBCQ $q$, it holds that
$D \cup \Sigma \not\models q$ implies that there exists a finite
model $M$ of $D \cup \Sigma$ such that $M \not\models q$.
This is formally stated as $D \cup \Sigma \models q$ if and 
only if $D \cup \Sigma \fmodels q$, or equivalently 
$\mathit{chase}(D, \Sigma) \models q$ if and 
only if $D \cup \Sigma \fmodels q$.

\new{
\subsection{Datalog$^\pm$ fragments}

Fix a database $D$, an ontology $\Sigma$, and a chase step 
involving some pair
$\langle \bar{\rho}, h\rangle$. 
To lighten the presentation, we assume that 
different rules of $\Sigma$ share no variable.  
Also, for every $m$-ary predicate $p$ and every
$i \in \{1,\ldots,m\}$, the pair $(p,i)$
is called {\em position} and denoted by $p[i]$.
Finally, given a set $A$ of atoms, a term $t$ occurs in $A$ at position $p[i]$
if there is $\alpha \in A$ s.t. $pred(\alpha) = p$ and $\alpha[i] = t$.

\myPar{Local conditions} 
%
%
$\Sigma$ belongs to:
$(i)$ $\mathsf{datalog}$ whenever  $\rho \in \Sigma$ implies \mbox{$\EV(\rho) = \emptyset$;} 
$(ii)$ $\mathsf{inclusion}$-$\mathsf{dependencies}$ whenever  $\rho \in \Sigma$ implies that
$\rho$ contains only simple atoms and $|body(\rho)| =1$;
$(iii)$ $\mathsf{linear}$ whenever  $\rho \in \Sigma$ implies $|body(\rho)| =1$; 
$(iv)$ $\mathsf{guarded}$ whenever  $\rho \in \Sigma$ implies that there is an atom 
of $body(\rho)$ containing all the variables of $\UV(\rho)$;
$(v)$ $\mathsf{joinless}$ whenever  $\rho \in \Sigma$ implies that 
$head(\rho)$ is a simple atom and $body(\rho)$ contains no repeated variables.

\myPar{Weak-acyclicity~\cite{DBLP:journals/tcs/FaginKMP05}}
Informally, $\Sigma \in \mathsf{weakly}$-$\mathsf{acyclic}$ guarantees that:
if $X$ occurs in $body(\bar{\rho})$ at position $p[i]$ and $h(X) \in {\bf N}$, then
the number of distinct nulls occurring in $\mathit{rchase}(D,\Sigma)$ 
at position $p[i]$ are finitely many.
Formally, the labeled graph $G(\Sigma)$ associated to $\Sigma$ is defined as the
pair $\langle N,A\rangle$, where 
$(i)$ $N$ collects all the positions $p[1],\ldots,p[m]$
for each $m$-ary predicate $p$ occurring in $\Sigma$;
$(ii)$ $(p[i], r[j], \texttt{plain}) \in A$ if there is a rule $\rho \in \Sigma$ and a variable $X$ of $\rho$ such that: $X$ occurs in the body of $\rho$ at position $p[i]$ and $X$ occurs in the head of $\rho$ at position $r[j]$; and
$(iii)$ $(p[i], r[j], \texttt{special}) \in A$ if there is a rule $\rho \in \Sigma$, a universal variable $X$ occurring also in the head of $\rho$, and an existential variable $Z$ of $\rho$ such that: $X$ occurs in the body of $\rho$ at position $p[i]$ and $Z$ occurs in the head of $\rho$ at position $r[j]$.
Ontology $\Sigma$ belongs to $\mathsf{weakly}$-$\mathsf{acyclic}$ if $G(\Sigma)$ 
has no cycle going through an arc labeled as $\texttt{special}$.

\myPar{Stickiness~\cite{DBLP:journals/ai/CaliGP12}}
Informally, $\Sigma \in \mathsf{sticky}$ guarantees that:
if $X$ occurs multiple times in $body(\bar{\rho})$, then 
$X$ occurs in $head(\bar{\rho})$ and 
$h(X)$ belongs to every atom of $\mathit{chase}(D,\Sigma)$
that depends on $h(head(\bar{\rho}))$.
Formally, a variable $X$ of $\Sigma$ is {\em marked} if 
$(i)$ there is a rule $\rho\in\Sigma$ 
such that $X$ occurs in $body(\rho)$ but not in $head(\rho)$; or
$(ii)$ there are two rules $\rho, \rho' \in\Sigma$
such that a marked variable occurs in $body(\rho)$ at some position $p[i]$ and
$X$ occurs in $head(\rho')$ at position $p[i]$ too.
Ontology $\Sigma$ belongs to  $\mathsf{sticky}$ if, for each $\rho \in \Sigma$, the following condition is satisfied: if $X$ occurs multiple times in $body(\rho)$, then $X$ is not marked.
A more refined condition identifies interesting cases 
in which it is safe to allow rules containing some marked variable 
that occurs multiple times but in a single body atom only.
This refinement gives rise to $\mathsf{sticky}$-$\mathsf{join}$, 
generalizing both $\mathsf{sticky}$ and $\mathsf{linear}$.

\myPar{Shyness~\cite{DBLP:conf/kr/LeoneMTV12}}
Informally, $\Sigma \in \shy$ guarantees that:
$(1)$ if $X$ occurs in two different atoms of $body(\bar{\rho})$, then $h(X) \in {\bf C}$; and
$(2)$ if $X$ and $Y$ occur both in $head(\bar{\rho})$ and in two different atoms of $body(\bar{\rho})$, then
$h(X) = h(Y)$ implies $h(X) \in {\bf C}$.
Formally, 
consider an existential variable $X$ of $\Sigma$.
Position $p[i]$ is {\em invaded} by $X$ if there is a rule $\rho$ of $\Sigma$ such
that:
$(i)$ $X$ occurs in $head(\rho)$ at position $p[i]$, or
$(ii)$ some universal variable $Y$ of $\rho$ is {\em attacked} by $X$ ---namely 
$Y$ occurs in $body(\rho)$ only at positions invaded by $X$--- and
it also occurs in $head(\rho)$ at position $p[i]$.
%
%
%
%
A universal variable is \textit{protected}
if it is attacked by no existential variable.
Ontology $\Sigma$ belongs to $\shy$ if, for each $\rho \in \Sigma$, the following conditions are both satisfied:
$(1)$ if $X$ occurs in two different atoms of $body(\rho)$, then $X$ is protected; and
$(2)$ if $X$ and $Y$ occur both in $head(\rho)$ and in two different atoms of $body(\rho)$, then
$X$ and $Y$ are not attacked by the same variable.}

\beforeSec
\section{Canonical rewriting}\label{sec:tools}


In this section we design a general technique 
to facilitate the process of (dis)proving finite controllability
of an arbitrary ontological fragment of existential rules.
%
%
More specifically, from a triple $\langle D,\Sigma,q\rangle $ we build
the triple $\langle D^c,\Sigma^c,q^c\rangle$ enjoying the following properties: 
$(1)$ $D^c$ is propositional database;
$(2)$  $\Sigma^c$ are constant-free rules containing only simple atoms;
$(3)$ $q^c$ is a constant-free UBCQ with only simple atoms;
$(4)$ $\mathit{chase}(D^c,\Sigma^c)$ is a constant-free instance
containing only simple atoms; and
$(5)$ there is a ``semantic'' correspondence between $\mathit{mods}(D,\Sigma)$ and $\mathit{mods}(D^c,\Sigma^c)$.
By exploiting these properties, one can 
apply the technique shown in Figure~\ref{fig:chain1}.

%
%
%
%
%
%
%
%

\begin{figure}[t!]
	\begin{center}
		\begin{tikzpicture}
		[->,shorten >=1pt,auto,thick,main node/.style={font=\sffamily\bfseries}]
		
		\node[main node] (1) {$D\cup\Sigma\models q$};
		
		\node[main node] (2) [right of=1,node distance=4.5cm] 
		{$D\cup\Sigma\fmodels q$};
		
		\node[main node] (3) [below of=2,node distance=1.4cm] 
		{$D^c\cup \Sigma^c \fmodels q^c$};
		
		\node[main node] (4) [below of=1,node distance=1.4cm] 
		{$D^c\cup\Sigma^c \models q^c$};

		\node[main node] (11) [right of=2,node distance=2.5cm] {$D\cup\Sigma\models q$};
		
		\node[main node] (21) [right of=11,node distance=4.5cm] 
		{$D\cup\Sigma\fmodels q$};
		
		\node[main node] (31) [below of=21,node distance=1.4cm] 
		{$D^c\cup \Sigma^c \fmodels q^c$};
		
		\node[main node] (41) [below of=11,node distance=1.4cm] 
		{$D^c\cup\Sigma^c \models q^c$};

		\draw[->] (2) to[right] node [below] {\scriptsize to prove this...} (1);
		
		\draw[dashed] (2) to node [left] {\scriptsize Th.~\ref{thm:fin-to-fin}} (3);
		
		\draw[dashed] (4) to node [right] {\scriptsize Th.~\ref{thm:inf-to-inf}} (1);
		
		\draw[dashed] (3) to[left] node [above] {\scriptsize ...show this is true} (4);

		\draw[->] (21) to[right] node [below] {\scriptsize to disprove this...} (11);
		
		\draw[dashed] (31) to node [left] {\scriptsize Th.~\ref{thm:fin-to-fin}} (21);
		
		\draw[dashed] (11) to node [right] {\scriptsize Th.~\ref{thm:inf-to-inf}} (41);
		
		\draw[dashed] (31) to[left] node [above] {\scriptsize ...show this is false} (41);
		
		\end{tikzpicture}
	\end{center}
	\beforeCaption
	\caption{Application of the canonical rewriting.}\label{fig:chain1}
	\afterCaption
\end{figure}
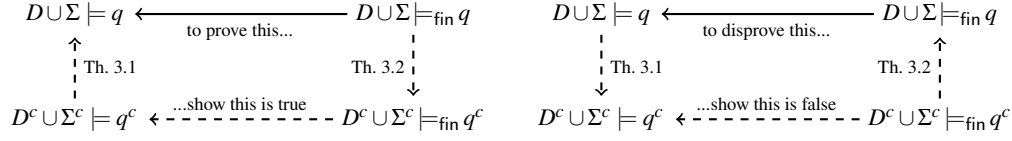

\beforeSubsec

\subsection{Overview}\label{sec:overview}

%
Consider the database $D=\{\mathit{person}(\mathit{tim})$, 
$\mathit{person}(\mathit{john})$, 
$\mathit{fatherOf}(\mathit{tim},\mathit{john})\}$, and the
ontology
$\Sigma=$
$\{\mathit{person}(X) \rightarrow \exists Y \, \mathit{fatherOf}(Y,X)$; 
$\mathit{fatherOf}(X,Y) \rightarrow \mathit{person}(X)\}$. 
Let $p$, $f$, $c_1$ and $c_2$ be shorthands of 
$\mathit{person}$, $\mathit{fatherOf}$, $\mathit{tim}$ and $\mathit{john}$, respectively.
Hence, $\mathit{chase}(D,\Sigma)$ is the instance 
$D\cup \{f(n_1,c_1), f(n_2,c_2)\}
\cup
\{p(n_i), f(n_{i+2},n_i) \}_{i>0}$\new{, where each $n_i$ denotes a distinct null of {\bf N}}.
From $D$ we construct the propositional database 
$D^c = \{p_{[c_1]}$, $p_{[c_2]}$, $f_{[c_1,c_2]}\}$ 
obtained by encoding in the predicates the tuples of $D$.
Then, from $\Sigma$ we construct $\Sigma^c$ collecting the following rules:

\befAftVectors

\newcommand{\evid}[1]{\textcolor{red}{#1}}
{\small
\[\begin{array}{rclrclrclrcl}
\evid{p_{[c_1]}} &\hspace*{-0.2cm} \evid{\rightarrow} &\hspace*{-0.2cm} \evid{\exists Y f_{[1,c_1]}(Y)} & \ \ \ \ \
f_{[c_1,c_1]} &\hspace*{-0.2cm} \rightarrow &\hspace*{-0.2cm} p_{[c_1]} &  \ \ \ \ \
f_{[c_1,1]}(Y) &\hspace*{-0.2cm} \rightarrow &\hspace*{-0.2cm} p_{[c_1]} &  \ \ \ \ \
f_{[1,1]}(X) &\hspace*{-0.2cm} \rightarrow &\hspace*{-0.2cm} p_{[1]}(X)\\
\evid{p_{[c_2]}} &\hspace*{-0.2cm} \evid{\rightarrow} &\hspace*{-0.2cm} \evid{\exists Y f_{[1,c_2]}(Y)} &
\evid{f_{[c_1,c_2]}} &\hspace*{-0.2cm} \evid{\rightarrow} &\hspace*{-0.2cm} \evid{p_{[c_1]}} &
f_{[c_2,1]}(Y) &\hspace*{-0.2cm} \rightarrow &\hspace*{-0.2cm} p_{[c_2]} &
\evid{f_{[1,2]}(X,Y)} &\hspace*{-0.2cm} \evid{\rightarrow} &\hspace*{-0.2cm} \evid{p_{[1]}(X)}\\
\evid{p_{[1]}(X)} &\hspace*{-0.2cm} \evid{\rightarrow} &\hspace*{-0.2cm} \evid{\exists Y f_{[1,2]}(Y,X)} &
f_{[c_2,c_1]} &\hspace*{-0.2cm} \rightarrow &\hspace*{-0.2cm} p_{[c_2]} & 
\evid{f_{[1,c_1]}(X)} &\hspace*{-0.2cm} \rightarrow &\hspace*{-0.2cm}  \evid{p_{[1]}(X)}
& & &\\
 &   &  &
f_{[c_2,c_2]} &\hspace*{-0.2cm} \rightarrow &\hspace*{-0.2cm} p_{[c_2]} &
\evid{f_{[1,c_2]}(X)} &\hspace*{-0.2cm} \evid{\rightarrow} &\hspace*{-0.2cm} \evid{p_{[1]}(X)}
& & &
\end{array}\]
}

\befAftVectors

\noindent The predicates here encode tuples of terms consisting of database constants ($c_1$ and $c_2$)
and placeholders of nulls ($1$ and $2$).
Consider the first rule $\rho = p(X) \rightarrow \exists Y f(Y,X)$  
applied by the chase over $D\cup \Sigma$,
and $h = \{X \mapsto c_1, Y \mapsto n_1\}$ be its associated homomorphism.
Hence, $h(\mathit{body}(\rho))  = p(c_1)$ and $h(\mathit{head}(\rho)) = f(n_1,c_1)$.
Such an application is mimed by the ``sister'' rule
$\rho^c = p_{[c_1]} \rightarrow \exists Y f_{[1,c_1]}(Y)$. By exploiting the same
homomorphism we obtain 
$h(\mathit{body}(\rho^c))  = p_{[c_1]}$ and also $h(\mathit{head}(\rho^c)) = f_{[1,c_1]}(n_1)$.
Actually, the encoded tuple $[c_1]$ in $p_{[c_1]}$ 
says that the original twin atom $p(c_1)$ is unary 
and its unique term is exactly $c_1$.
Moreover, the encoded tuple $[1,c_1]$ in $f_{[1,c_1]}(n_1)$
says that the original twin atom $f(n_1,c_1)$ is binary,
that its first term is a null, and that its second term is exactly the constant $c_1$.
Since from predicate $f_{[1,c_1]}$ we only know that the first term is a null, 
it must be unary to keep the specific null value.
In the above construction, red rules are those applied by the chase on $D^c \cup \Sigma^c$.
For example, 
rule $f_{[1,c_1]}(X) \rightarrow p_{[1]}(X)$ mimics
$f(X,Y) \rightarrow p(X)$ when $X$ is mapped to a null and
$Y$ to $c_1$; and
rule $f_{[1,2]}(X,Y) \rightarrow p_{[1]}(X)$
mimics 
$f(X,Y) \rightarrow p(X)$ when $X$ and $Y$ are mapped to different nulls.
Hence, $\mathit{chase}(D^c,\Sigma^c)$ is:
$$D^c \cup
\{p_{[1]}(n_i)\}_{i>0}
\cup
\{f_{[1,c_1]}(n_1), f_{[1,c_2]}(n_2)\}
\cup
\{f_{[1,2]}(n_{i+2},n_i)\rem{, f_{[1,2]}(n_{i+3},n_{i+1})}\}_{i>0}.$$
%
%
%

As a result, the rewriting separates the interaction between the database constants
propagated body-to-head via universal variables and the nulls
introduced to satisfy existential variables. 
Also, since the 
predicates encode the ``shapes'' of the twin atoms ---namely $f_{[1,2]}(X,Y)$ means different nulls
while $f_{[1,1]}(X)$ the same null---
repeated variables are encoded too.
%
%
By following the same approach, we can rewrite also the query. 
Consider for example the BCQ $q=\exists X p(X), f(X,c_1)$.
%
%
Therefore, $q^c$ is the UBCQ:
$ (p_{[c_1]}, \, f_{[c_1,c_1]}) \ \vee \ (p_{[c_2]}, \, f_{[c_2,c_1]}) \ \vee \ 
(\exists X \, p_{[1]}(X), \, f_{[1,c_1]}(X)).$


\beforeSubsec

\subsection{Formal construction and properties}

Let us fix a triple $\langle D,\Sigma,q\rangle$ through the rest of this section.
%
Consider an atom $\alpha = p(t_1,\ldots,t_m)$ with terms over $\C\cup\V$.
The {\em canonical atom} of $\alpha$ is the atom $\alpha^c=p_{[\ell_1,\ldots,\ell_m]}(\tau_1,\ldots,$ $\tau_\mu)$, where:
$(a)$ $\ell_i = t_i$ if $t_i\in\C$;
$(b)$ $\ell_i = \ell_j$ if $t_i=t_j$; or
$(c)$ $\ell_i = 1+\mbox{max}(\{0\}\cup\{\ell_j\in\mathbb{N}: j<i\})$ if $t_i\in\V$  and $t_j\neq t_i \ \forall j<i$\rem{.}
%
and $\tau_i = V\in\V$, if there exists $t_j$ such that $\ell_j=i$ and $t_j=V$.
%
Moreover, given a set of atoms $A$,
we define $A^c=\{\alpha^c : \alpha\in A\}$, and give a rule $\rho$, we define $\rho^c$ as the rule so that $body(\rho^c)=body(\rho)^c$ and $head(\rho^c)=head(\rho)^c$.
%
For instance, let $\alpha=p(c_1,X,c_2,X,Y,Z,Y)$ be an atom. Then, the canonical atom $\alpha^c$ of $\alpha$ is given by $p_{[c_1,1,c_2,1,2,3,2]}(X,Y,Z)$.
Note that, by definition of $\tau_i$, for $i=1,\ldots,\mu$, we have that
the arity $\mu \leq m$ of the canonical atom is equal to $\mbox{max}(\{0\}\cup\{f(t_j)\in\mathbb{N}: j\leq m\})$.

\rem{For each rule $\rho$ of $\Sigma$, we denote by ${\bf V}(\rho)$ the set of variables appearing in $\rho$, by $\EV(\rho)$ the set of all existential variables in $\rho$, and
by $\UV(\rho)$ the set of all universal variables in $\rho$. 
E.g., if $\rho$ is the rule $p(X,Y),r(Y)\rightarrow \exists Z \, s(X,Y,Z)$, then $\V(\rho)=\{X,Y,Z\}$, $\EV(\rho)=\{Z\}$, and $\UV(\rho)=\{X,Y\}$.}

\begin{definition}[Safe and Canonical substitutions]
A map $\varsigma:\const(D\cup\Sigma)\cup\V\rightarrow \const(D\cup\Sigma)\cup\V$ is called \textit{canonical substitution} 
if $\varsigma(c)=c$ for each $c\in\const(D\cup\Sigma)$.
Moreover, we say that a canonical substitution $\varsigma$ is \textit{safe} w.r.t. a rule $\rho\in\Sigma$ if 
$\varsigma(\UV(\rho)) \subseteq \const(D\cup\Sigma)\cup\UV(\rho)$, and
$\varsigma(V)=V$, for each $V\in\EV(\rho)$.
\end{definition}

Intuitively, a safe substitution maps each existential variable \rem{in}\new{to} itself and no universal variable is mapped to an existental one.
As usual, given a set of atom\new{s} $A$, we denote by $\varsigma(A)=\{p(\varsigma(t_1),\ldots,\varsigma(t_m)\new{)}: p(t_1,\ldots,t_m)\}$, and given a rule $\rho$, we denote by $\varsigma(\rho)$ the rule such that $body(\varsigma(\rho))=\varsigma(body(\rho))$ and
$head(\varsigma(\rho))=\varsigma(head(\rho))$.

\begin{example}\label{ex:safe-sub}
Consider $D=\{  r(c_1,c_3)  \}$ and $\Sigma$ consisting of the following rules:
$\rho_1 =  r(Y_1,Z_1),p(W_1,X_1,X_1,Y_1)$ 
$ \rightarrow 
 \exists T_1 g(X_1,Y_1,T_1,X_1,Z_1)$ and
$\rho_2 =   s(X_2),t(Y_2) 
\rightarrow 
r(X_2,Y_2)$.
For instance, $\varsigma_1=$ $\{c_1\mapsto c_1, c_3\mapsto c_3, Y_1\mapsto X_1,$ $Z_1\mapsto c_3,$ $W_1\mapsto Y_1,$ $X_1\mapsto X_1,$ $T_1\mapsto T_1\}$ and
$\varsigma'_1=$ $\{c_1\mapsto c_1, c_3\mapsto c_3, Y_1\mapsto c_1,$ $Z_1\mapsto X_1,$ $W_1\mapsto c_1,$ $X_1\mapsto Y_1,$ $T_1\mapsto T_1\}$
are safe substitutions w.r.t. $\rho_1$.
Indeed, $\const(D\cup\Sigma)=\{c_1,c_3\}$,
$\UV(\rho_1)=\{W_1,X_1,Y_1,Z_1\}$,
$\EV(\rho_1)=\{ T_1\}$, the existential variable $T_1$ is mapped to itself, and no other variable is mapped to an existential one. Moreover, $\varsigma_1(\rho_1)=r(X_1,c_3),p(Y_1,X_1,X_1,X_1)\rightarrow \exists T_1 g(X_1,X_1,T_1,X_1,c_3)$ and $\varsigma'_1(\rho_1)=r(c_1,X_1),p(c_1,Y_1,Y_1,c_1)\rightarrow \exists T_1 g(Y_1,c_1,T_1,Y_1,X_1)$.
\hfill $\lhd$
\end{example}

\rem {Note that, if $\rho$ is a $\mathsf{datalog}$ rule (for instance, rule $\rho_2$ in the previous example), then each canonical substitution 
is a safe substitution w.r.t. $\rho$, as $\EV(\rho)=\emptyset$.}
We denote by $\CS$ the set of all canonical substitutions and by
$\ss(\rho)\subseteq \CS $ the set of all safe substitutions w.r.t. $\rho$.
Given a set of atoms $A$ [resp. a rule $\rho$] and a canonical substitution [resp. safe substitution] $\varsigma$, we say that
$\varsigma(A)^c$ [resp. $\varsigma(\rho)^c$] is the 
\textit{canonical set of atoms} w.r.t $A$
[resp. \textit{canonical rule w.r.t. $\rho$}] and $\varsigma$.
Observe that two different canonical substitutions could produce two isomorphic canonical set of atoms. For instance, let $A = \{p(X,Y)\}$, and consider
$\varsigma =\{X\mapsto X, \ Y\mapsto Y\}$ and
$\varsigma' =\{X\mapsto Y, \ Y\mapsto X\}$. Then,
$\varsigma(A)^c = \{p_{[1,2]}(X,Y)\}$, and
$\varsigma'(A)^c = \{p_{[1,2]}(Y,X)\}$ are isomorphic set of atoms. 
%
%
\rem{Therefore, we denote by $\CS^{*}$ [resp. $\ss^{*}(\rho)$] a maximal set of canonical substitutions [resp. safe substitutions w.r.t. $\rho$] that produce canonical set of atoms [resp. canonical rules] pairwise non-isomorphic.}%
\new{Therefore, to avoid redundancies, we denote by $\CS^{*}$ [resp. $\ss^{*}(\rho)$] 
any arbitrary maximal subset of $\CS$ [resp. of $\ss(\rho)$] 
producing canonical set of atoms [resp. canonical rules] 
containing no two isomorphic elements.}
%

We denote by $\Sigma^c$ the set of all canonical rules 
$\{ \varsigma(\rho)^c: \rho\in\Sigma \mbox{ and } \varsigma\in\ss^{*}(\rho) \}$,
and we call it the \textit{canonical rewriting} of $\Sigma$. 
Also, given a UBCQ $q$ of the form
$\exists {\bf Y}_1 \psi_1({\bf Y}_1) \vee \ldots \vee \exists {\bf Y}_k \psi_k({\bf Y}_k)$,
we denote by $q^c$ the disjunction
$\bigvee_{\varsigma_1 \in \CS^{*}} 
\varsigma_1(\psi_1({\bf Y}_1))^c \vee \ldots \vee 
\bigvee_{\varsigma_k \in \CS^{*}} 
\varsigma_k(\psi_k({\bf Y}_k))^c$.
and we call it the \textit{canonical rewriting} of $q$. 
Finally, we call $D^c$ the {\em canonical rewriting} of $D$.
%
%

\new{
\begin{proposition}
	The triple $\langle D^c,\Sigma^c,q^c\rangle$ can be constructed from $\langle D,\Sigma,q\rangle$ in 
	polynomial time (in data complexity).
\end{proposition}
}

\begin{example}\label{ex:can-rew}
Consider the ontology $\Sigma$ with the safe substitutions $\varsigma_1$ and $\varsigma'_1$ w.r.t. $\rho_1$ of the Example~\ref{ex:safe-sub}.
Therefore, we obtain the canonical rules:
$\varsigma_1(\rho_1)^c=  
r_{[1,c_3]}(X_1),p_{[1,2,2,2]}(Y_1,X_1) \rightarrow \exists T_1 g_{[1,1,2,1,c_3]}(X_1,T_1)$ and
$\varsigma'_1(\rho_1)^c =
r_{[c_1,1]}(X_1),p_{[c_1,1,1,c_1]}(Y_1) \rightarrow  \exists T_1 g_{[1,c_1,2,1,3]}(Y_1,T_1,X_1)$.
Moreover, let $\varsigma_2$ and $\varsigma'_2$ be the safe substitutions containing
$\{X_2\mapsto X_2, Y_2\mapsto X_2\}$ 
and
$\{X_2\mapsto c_1, Y_2\mapsto c_3\}$ w.r.t. $\rho_2$, respectively.
Hence, we have
$\varsigma_2(\rho_2)^c=  s_{[1]}(X_2),t_{[1]}(X_2) \rightarrow r_{[1,1]}(X_2)$ and
$\varsigma'_2(\rho_2)^c= s_{[c_1]},t_{[c_3]}\rightarrow r_{[c_1,c_3]}$.
Therefore, $\varsigma_1(\rho_1)^c$,
$\varsigma'_1(\rho_1)^c$,
$\varsigma_2(\rho_2)^c$, and
$\varsigma'_2(\rho_2)^c$ are (some of the) rules of $\Sigma^c$. 
\hfill $\lhd$
\end{example}

We consider a function $\rp$ from
the set of atoms of $D^c\cup\Sigma^c$ to the set of atom of $D\cup\Sigma$.
For each atom $\alpha=a_{[s_1,\ldots,s_m]}(\sigma_1,\ldots,\sigma_\mu)$, we 
build an atom $\rp(\alpha)=a(t_1,\ldots,t_m)$ such that:
$(a)$ $t_i = s_i$ if  $s_i\in{\bf C}$;
$(b)$ $t_i = \sigma_i$ if $s_i=k$  and $s_j\neq k$, for each $j<i$; or
$(c)$ $t_i = \sigma_j$ if $s_i=s_j$, for some $j<i$.

%
For instance, let $\alpha=p_{[1,c_1,2,1,c_2,1,2]}(X,Y)$ be an atom of the logical theory $D^c\cup\Sigma^c$. Then, $\rp(\alpha)=p(X,c_1,Y,X,c_2,X,Y)$.
We call $\rp$ the \textit{unpacking function}.
Given a set of atoms $A$ of $D^c\cup\Sigma^c$, we denote by
$\rp(A)=\{ \rp(\alpha) : \alpha\in A\}$ the corresponding set of atoms of $D\cup\Sigma$.
If $I$ is an instance, we call $\rp(I)$ the \textit{unpacked instance} of $I$.
Given a rule $\rho^c$ in $\Sigma^c$, we denote by
$\rp(\rho^c)$ the rule obtained applying $\rp$ to each atom in $\rho^c$,
i.e. $\rp(\rho^c): \rp(body(\rho^c))\rightarrow \rp(head(\rho^c))$, 
and we call it the \textit{unpacked rule} of $\rho^c$.
Similarly, we denote by $\rp(q^c)$ the query obtained applying $\rp$ to the atoms of the UBCQ $q^c$, and we call it the \textit{unpacked query} of $q^c$.
%
%
%
%
Informally, the unpacking function acts as the inverse operator to the canonical rewriting.
Moreover, it enjoys an interesting and useful property: the chase of a logical theory 
coincides with the unpacking of the chase constructed from of the same theory given in canonical form:
\rem{This is formally stated in the following proposition.}

\begin{proposition}\label{th:inf-to-inf}
Consider a set $\Sigma$ of existential rules.
For each database $D$ and for each UBCQ $q$, it holds that $\rp(\mathit{chase}(D^c,\Sigma^c)) = \mathit{chase}(D,\Sigma)$ and $\rp(q^c) \equiv q$.
\end{proposition}

By exploiting the above proposition, we can now prove that a UBCQ $q$ is satisfied by all models of a theory $D\cup\Sigma$ if, and only if, each model of the canonical rewriting of the theory $D^c\cup\Sigma^c$ satisfies the canonical rewriting of the UBCQ $q^c$.


\begin{theorem}\label{thm:inf-to-inf}
$D\cup\Sigma\models q$ if, and only if, $D^c\cup\Sigma^c\models q^c$.
\end{theorem}


Note that, if $\Sigma$ is a constant-free ontology, then, for each model $M^c$ of $D^c\cup\Sigma^c$, $\rp(M^c)$ is a model of $D\cup\Sigma$.  
The request for a constant-free ontology is needed. 
Indeed, for instance, let $\Sigma=\{p(a) \rightarrow r(a); \ r(x) \rightarrow p(x)\}$. 
So that,
$\Sigma^c=\{p_{[a]} \rightarrow r_{[a]}; \ r_{[a]} \rightarrow p_{[a]};
\ \ r_{[1]}(V_1) \rightarrow p_{[1]}(V_1)\}$. 
Therefore, $M^c=\{p_{[1]}(a)\}$ is a model of $\Sigma^c$, but $\rp(M^c)=\{p(a)\}$ is not a model of $\Sigma$, as it does not satisfy the first rule.
However, we can overcome this problem considering the following class of models.
Given a model $M^c\in\mods(D^c,\Sigma^c)$, we say that $N^c$ is a \textit{\smooth}  instance of $M^c$ if 
there exists a bijective map $f:\terms(M^c)\rightarrow \terms(N^c)$ such that 
$f(n)=n$ for each null $n\in\terms(M^c)$;
$f(c)=n_c$ for each constant $c\in\terms(M^c)$, where $n_c$ is a fresh null;
and $f(M^c)=N^c$.
Note that a {\smooth} instance of a model $M^c$ is also a model of $D^c\cup\Sigma^c$ and it is also constant-free.
\begin{proposition}\label{th:fin-to-fin}
If $M^c\in\mods(D^c,\Sigma^c)$,
then $\rp(N^c)\in\mods(D,\Sigma)$, for each {\smooth} model $N^c$ of $M^c$.
\end{proposition}

By exploiting the above proposition, we can now prove that a UBCQ $q$ is satisfied by all finite models of a theory $D\cup\Sigma$ if, and only if, each finite model of the canonical rewriting of the theory $D^c\cup\Sigma^c$ satisfies the canonical rewriting of the UBCQ $q^c$.

\begin{theorem}\label{thm:fin-to-fin}
	$D\cup\Sigma\models_\mathsf{fin} q$ if, and only if, $D^c\cup\Sigma^c\models_\mathsf{fin} q^c$.
\end{theorem}

\beforeProof

\begin{proof}
Assume that $D\cup\Sigma \fmodels q$. 
Then, for each finite model $M$ of $D\cup\Sigma$, there exists a homomorphism $h$ from at least one disjunct of $q$, say $\psi_j(\Y_j)$ to $M$.
Now, let $M^c$ be a finite model of $D^c\cup\Sigma^c$.
By Proposition~\ref{th:fin-to-fin}, there exist a (finite) {\smooth} model $N^c\in\mods(D^c,\Sigma^c)$ of $M^c$ and a bijective map
$f$ from $terms(M^c)$ to $terms(N^c)$ such that 
$f(n)=n$ for each null $n\in\terms(M^c)$;
$f(c)=n_c$ for each constant $c\in\terms(M^c)$, where $n_c$ is a fresh null;
$f(M^c)=N^c$, and $\rp(N^c)\in\mods(D,\Sigma)$.
Hence, by assumption, there exists a homomorphism $h$ from some $\psi_j(\Y_j)$ to $\rp(N^c)$.
Let $A=h(\psi_j(\Y_j))\subseteq \rp(N^c)$.
Then, for each atom $\alpha\in A$, we can choose an arbitrary atom $\beta\in N^c$ such that $\rp(\beta)=\alpha$. 
Let $B$ such a subset of $N^c$.
Therefore, by construction, there exists a BCQ in $q^c$ isomorphic to $B$.
In particular, there exists a homomorphism $h$ from $q^c$ to $N^c$.
In conclusion, $f^{-1}\circ h$ is a homomorphism from $q^c$ to $M^c$. 
Indeed, $f^{-1}\circ h$ is a map from $terms(q^c)$ to $\terms(M^c)$ such that
$f^{-1}(h(q^c))\subseteq f^{-1}(N^c)=M^c$.
Now, assume that $D^c\cup\Sigma^c\fmodels q^c$.
Let $M$ be a finite model of $D\cup\Sigma$.
By definition of canonical rules, can be easily proved that
there exists a finite model $M^c\in\fmods(D^c,\Sigma^c)$ such that $\rp(M^c)=M$.
Hence, let $h$ be a homomorphism from $q^c$ to $M^c$.
So that, $h(\varsigma_j(\psi_j(\Y_j))^c)\subseteq M^c$, for some disjunct $\varsigma_j(\psi_j(\Y_j))^c$ of $q^c$.
Therefore, by applying the unpacked function, we have that
$\rp(h(\varsigma_j(\psi_j(\Y_j))^c))=h(\rp(\varsigma_j(\psi_j(\Y_j))^c))=
h(\varsigma_j(\psi_j(\Y_j))\subseteq \rp(M^c)=M$.
Hence, $h$ is also a homomorphism from $q$ to $M$.
\end{proof}

\beforeSubsec

\subsection{Immediate consequences}

\rem{By exploiting the properties of the canonical rewriting, one can immediately observe 
	(actually reprove) that $\mathsf{linear}$ is finitely controllable since
	$\mathsf{inclusion}$-$\mathsf{dependencies}$ is finitely controllable.
	In fact, given a $\mathsf{linear}$ ontology $\Sigma$,
	its canonical rewriting $\Sigma^c$ belongs to $\mathsf{inclusion}$-$\mathsf{dependencies}$.	
	However, it is also possible to prove (now for the first time)
	that $\mathsf{sticky}$-$\mathsf{join}$ is finitely controllable since
	$\mathsf{sticky}$ is finitely controllable~\cite{DBLP:conf/lics/GogaczM13}.
	Actually, $\mathsf{sticky}$-$\mathsf{join}$ extends both $\mathsf{sticky}$
and $\mathsf{linear}$ by admitting non simple atoms even when the 
stickyness property is violated~\cite{DBLP:journals/ai/CaliGP12}. 
And one can show that after applying the canonical rewriting
to a $\mathsf{sticky}$-$\mathsf{join}$ ontology what we obtain is
a $\mathsf{sticky}$ ontology. The following results follows.}%
\new{
By exploiting the properties of the canonical rewriting, one can 
reprove that $\mathsf{linear}$ is finitely controllable,
and prove (for the first time)
that also $\mathsf{sticky}$-$\mathsf{join}$ enjoys this property.
In fact, given a $\mathsf{linear}$ or $\mathsf{sticky}$-$\mathsf{join}$
ontology $\Sigma$, its canonical rewriting $\Sigma^c$ belongs to $\mathsf{inclusion}$-$\mathsf{dependencies}$ or $\mathsf{sticky}$, respectively.
In the former case, it suffices to observe that 
any variable occurring multiple times in some atom $\alpha$, by definition, 
occurs exactly once in its associated canonical atom $\alpha^c$.
In the latter case, additionally, 
consider a variable $X$ violating the sticky property since it is 
marked and it occurs multiple times in the body of some rule $\rho$. 
By hypothesis,
$X$ may occur in exactly one atom of $body(\rho)$. However, 
even if marked, $X$ now occurs exactly once in its canonical atom and 
it cannot violate the sticky property any more.
The following result follows.}

\begin{theorem}\label{th:stickyjoin}	
	Under $\mathsf{sticky}$-$\mathsf{join}$ ontologies, $D \cup\Sigma \models q$ if, and only if,
	$D\cup\Sigma\fmodels q$.
\end{theorem}


\beforeSec
\section{Finite controllability of Shy ontologies}\label{sec:challenge}

\begin{figure}[t!]
	\begin{center}
		\begin{tikzpicture}
		[->,shorten >=1pt,auto,thick,main node/.style={font=\sffamily\bfseries}]
		
		\node[main node] (1) {$D^c\cup\Sigma^c\models q^c$};
		\node[main node] (2) [right of=1,node distance=6.5cm] 
		{$D^c\cup\Sigma^c\fmodels q^c$};
		\node[main node] (3) [right of=2,node distance=4cm] 
		{$D^c\cup \Sigma^c \wsfmodels q^c$};
		\node[main node] (4) [below of=2,node distance=1.4cm] 
		{$D^c\cup\Sigma^c_a\fmodels q^c$};
		\node[main node] (6) [below of=3,node distance=1.4cm] {$D^c\cup\Sigma^c_a\wsfmodels q^c$};
		\node[main node] (5) [left of=4,node distance=6.5cm] {$D^c\cup\Sigma^c_a\models q^c$};		
		
		\draw (2) to[right] node [below] {\scriptsize Lemma. \ref{lm:main}} (1);
		\draw[dashed] (2) to node [below] {\scriptsize Th.~\ref{th:fsn-to-f}} (3);
		\draw[dashed] (6) to[left] node [above] {\scriptsize Th.~\ref{th:fsn-to-f} } (4);
		\draw[dashed] (4) to[left] node [above] {\scriptsize 
			\citeN{DBLP:conf/lics/GogaczM13}} (5);		
		\draw[dashed] (5) to node [right] {\scriptsize Th.~\ref{thm:activeCanonicalGeneral}} (1);
		\draw[dashed] (3) to node [left] {\scriptsize Th.~\ref{th:f-to-fsn}} (6);		
		\end{tikzpicture}
	\end{center}
	\beforeCaption
	\caption{Chain of implications for the proof of Lemma \ref{lm:main}.}\label{fig:chain}
	\afterCaption
\end{figure}
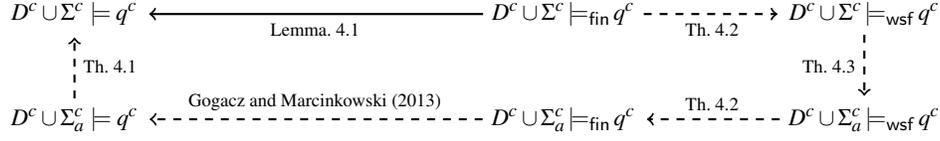

We open this section by observing that, differently from $\mathsf{linear}$ and 
$\mathsf{sticky}$-$\mathsf{join}$, 
the canonical rewriting of a $\mathsf{shy}$ ontology ---although it is still a $\mathsf{shy}$ ontology--- does not fall in any other known class.
To prove that $\shy$ is finitely controllable,
we therefore devise three technical tools on top of the
canonical rewriting defined in Section~\ref{sec:tools}.
These tools allow us to show that $D^c \cup\Sigma^c \fmodels q^c$ if, and only if,
$D^c\cup\Sigma^c\models q^c$ (Lemma~\ref{lm:main}).
%
%
To this end, let us fix a triple 
$\langle D,\Sigma,q \rangle$, and 
the associated one $\langle D^c,\Sigma^c,q^c \rangle$ in canonical form.
Our tools are as follows:

{\em Active and harmless rules.}
Whenever $\Sigma$ is $\shy$, we can partition $\Sigma^c$ in two sets, denoted by $\Sigma^c_a$ and $\Sigma^c_h$ ---collecting 
{\em active} and {\em harmless} rules, respectively--- enjoying the following properties:
$(1)$ $\Sigma^c_h$ are the rules of $\Sigma^c$ with
at least a variable occurring in more than one body atom;
$(2)$ $\Sigma^c_a = \Sigma^c\setminus \Sigma^c_h$ is a $\mathsf{joinless}$ (and still $\shy$) ontology; and
%
%
$(3)$ $\mathit{chase}(D^c,\Sigma^c) = \mathit{chase}(D^c,\Sigma^c_a)$.

{\em Well-supported finite models.}
Inspired by well-supported interpretations
of general logic programs~\cite{DBLP:journals/ngc/Fages91},
we define {\em well-supported} finite models 
of $\langle D,\Sigma\rangle $, denoted by $\mathit{wsfmods}(D,\Sigma)$, 
which enjoy the following properties: 
$(1)$ for each $M \in \mathit{wsfmods}(D,\Sigma)$, 
there exists an {\em ordering} $(\alpha_1,\ldots,\alpha_m)$
of its atoms such that, for each $\alpha_j$ of $M$,
either $\alpha_j$ belongs to $D$, or 
there exist a rule $\rho\in\Sigma$ and a homomorphism from the atoms of $\rho$ to $\{\alpha_1,\ldots,\alpha_{j}\}$
that maps $body(\rho)$ to $\{\alpha_1,\ldots,\alpha_{j-1}\}$
and $head(\rho)$ to $\{\alpha_j\}$;
$(2)$ for each $M \in \mathit{fmods}(D,\Sigma)$, there exists
a well-supported finite model $M'\subseteq M$; and
$(3)$ each minimal finite model of $D\cup\Sigma$ is a well-supported finite model.

{\em Propagation ordering.}
Since $\mathit{mods}(D^c,\Sigma^c) \subseteq \mathit{mods}(D^c,\Sigma^c_a)$,
in general it is definitely possible that a model $M$ of \mbox{$D^c \cup \Sigma^c_a$} 
is not a model of $D^c \cup \Sigma^c$.
In case $\Sigma$ is $\shy$ and $M$ is a well-supported finite models of $D^c \cup \Sigma^c_a$,
by exploiting an arbitrary ordering of $M$, we show how to rename and propagate some of 
the terms of $M$ to construct an instance $M'$ enjoying the
following property:
$(1)$ $M'\in \mathit{wsfmods}(D^c,\Sigma^c)$; and
$(2)$ there exists a homomorphism from $M'$ to $M$.

With these tools in place, we can now apply the technique shown in Figure~\ref{fig:chain},
where we use the symbol $\wsfmodels$ to refer the satisfiability of the query
under well-supported finite models only.


\beforeSubsec

\subsection{Active and harmless rules}\label{sec:active}

As said, and next stated, the canonical rewriting of a 
$\mathsf{shy}$ ontology is again a $\mathsf{shy}$ ontology.

\begin{proposition}\label{th:can-rew-of-shy-is-shy}
	If $\Sigma$ is $\shy$, then $\Sigma^c$ is. 
\end{proposition}

The goal of this section is therefore to identify a suitable subset of 
$\Sigma^c$ that falls in some known finitely-controllable class, 
and that roughly ``behaves'' as  $\Sigma^c$ under both finite and arbitrary models.
The idea is to collect in $\Sigma^c_h$ the rules of $\Sigma^c$ with
at least a variable occurring in more than one body atom, and 
to define $\Sigma^c_a = \Sigma^c\setminus \Sigma^c_h$.
In other words, $\Sigma^c_a$ is exactly the maximal subset of $\Sigma^c$
that belongs to $\mathsf{joinless}$.
%
%
%
%
Let us now provide some insights 
regarding this way of partitioning $\Sigma^c$.
From the database $D=\{p(c)\}$ and 
the $\shy$ ontology $\Sigma=$
$\{p(X) \rightarrow \exists Y \, f(Y,X);$
$f(X,Y), p(X) \rightarrow p(Y)\}$
we first construct $D^c = \{p_{[c]}\}$ 
and $\Sigma^c$ as the following set of rules:

\befAftVectors

{\small
	\[\begin{array}{rclrclrclr}
	f_{[1,c]}(X), p_{[1]}(X) & \rightarrow & p_{[c]} & \ \ \ \ \ 
	f_{[c,c]}, p_{[c]} & \rightarrow & p_{[c]} & \ \ \ \ \ 
	\evid{p_{[c]}} & \evid{\rightarrow} & \evid{\exists Y f_{[1,c]}(Y)} & \\
	f_{[1,1]}(X), p_{[1]}(X) & \rightarrow & p_{[1]}(X) & \ \ \ \ \ 
	f_{[c,1]}(Y), p_{[c]} & \rightarrow & p_{[1]}(Y) & \ \ \ \ \ 
	p_{[1]}(X) & \rightarrow & \exists Y f_{[1,2]}(Y,X) & \\
	f_{[1,2]}(X,Y), p_{[1]}(X) & \rightarrow & p_{[1]}(Y) & \ \ \ \ \ 
	 &  &  & \ \ \ \ \ 
	 &  &  & 	
	\end{array}
	\]}

\befAftVectors

\noindent Again, red rules are those applied by the chase on $D^c \cup \Sigma^c$.
Now we observe that there is no way to trigger the rules in the first column:
although the chase does produce an atom $f(t,n_i)$ for some term $t$ and null $n_i$,
it never produces any atom $p(n_i)$. 
This fact is detected by the syntactic conditions underlying $\shy$ 
(marking $X$ in $f(X,Y), p(X) \rightarrow p(Y)$ as ``protected''),
which guarantee that $X$ may be mapped by the chase to constants only.
%
%
Hence, since by definition $\Sigma^c_a$ consists of the joinless rules in the last two columns,
it holds that $\mathit{chase}(D^c,\Sigma^c) = \mathit{chase}(D^c,\Sigma^c_a)$.





The reason underlying the fact that the chase never applies rules of $\Sigma^c_h$ 
will be exploited in Section~\ref{sec:propOrd} to prove Theorem~\ref{th:f-to-fsn} (see Figure~\ref{fig:chain}, right-hand side),
namely that $\Sigma^c_a$ roughly ``behaves'' as  $\Sigma^c$ under finite (well-supported) models.
Conversely, to show Theorem~\ref{thm:activeCanonicalGeneral} (see Figure~\ref{fig:chain}, left-hand side) it suffices 
to observe the more general property that $\Sigma^c_a \subseteq \Sigma^c$,
which immediately implies $\mathit{mods}(D^c,\Sigma^c) \subseteq \mathit{mods}(D^c,\Sigma^c_a)$.
And the next result follows.

\begin{theorem}\label{thm:activeCanonicalGeneral}
	If $D^c\cup\Sigma^c_a\models q^c$, then $D^c\cup\Sigma^c\models q^c$.
\end{theorem}

\beforeSubsec

\subsection{Well-supported finite models}\label{sec:wsf}

We start by defining the notion of well-supported finite instances, 
which is inspired by the related notion of well-supported interpretations for general logic programs \cite{DBLP:journals/ngc/Fages91}.

%

Let $D$ be a database, and $\Sigma$ be an ontology.
A finite instance $I$ is called \textit{well-supported} 
w.r.t. the theory $D \cup \Sigma$ if there is an ordering $(\alpha_1,\ldots,\alpha_m)$ of its atoms
such that, for each $j \in \{1,\ldots,m\}$, at least one of the following conditions
is satisfied:
%
$(1)$ $\alpha_j$ is a database atom of $D$; and
$(2)$ there exist a rule $\rho$ of $\Sigma$ and a homomorphism $h$ from $atoms(\rho)$ to $\{\alpha_1,\ldots,\alpha_j\}$ such that $h(head(\rho))=\{\alpha_j\}$
    and  $h(body(\rho))\subseteq \{\alpha_1,\ldots,\alpha_{j-1}\}$.
In both cases, we will say that $\alpha_j$ is a \textit{well-supported atom} w.r.t. $(\alpha_1,\ldots,\alpha_m)$; while in the latter case we will also say that
$\rho$ is a \textit{well-supporting rule} for $\alpha_j$ w.r.t. $(\alpha_1,\ldots,\alpha_m)$.
Such an ordering will be called a \textit{well-supported ordering} of $I$.

We denote by $\mathit{wsfmods}(D,\Sigma) \subseteq \mathit{fmods}(D,\Sigma)$ the set of all well-supported finite models of $D \cup \Sigma$.
Moreover, if a UBCQ $q$ is satisfied by each model of $\mathit{wsfmods}(D,\Sigma)$, 
we write $D\cup\Sigma \models_{\mathsf{wsf}} q$.
Interestingly, each finite model of $D\cup\Sigma$ contains a well-supported finite model of the theory.

\begin{proposition}\label{th:wsfm-in-fm}
For each $M\in \mathit{fmods}(D,\Sigma)$, there exists $M'\subseteq M$
such that $M'\in \mathit{wsfmods}(D,\Sigma)$. In particular,
each minimal finite model of $D\cup\Sigma$ is a well-supported finite model.
\end{proposition}

%

Although each finite model of an ontological theory contains a well-supported finite model of the theory, the reverse inclusion does not hold.
Consider for example the ontology $\Sigma$ of Section~\ref{sec:overview}, and
the model $M=D\cup \{f(c_1,c_1),f(c_2,c_1)\}$.
Since $(p(c_1),p(c_2),f(c_1,c_2),f(c_1,c_1),$ $f(c_2,c_1))$ is a well-supported ordering of $M$, then $M$ is well-supported.
However, $M\setminus \{f(c_2,c_1)\}$ is a model of $D\cup\Sigma$.
Therefore, $M$ is not a minimal one.
Using Proposition~\ref{th:wsfm-in-fm}, we can now prove that if a UBCQ $q$ can be satisfied by each well-supported finite model of a theory, then it can be satisfied by each finite model of the theory.

\begin{theorem}\label{th:fsn-to-f}
	$D \cup\Sigma \models_\mathsf{wsf} q$ if, and only if,
	$D \cup \Sigma \models_\mathsf{fin} q$.
\end{theorem}

\beforeProof

\begin{proof}
Clearly, by subset inclusion, if each finite model of $D\cup\Sigma$ satisfies the query $q$, then each well-supported finite model of $D\cup\Sigma$ satisfies $q$.
Moreover,
	as each finite minimal model is a well-supported finite model (Proposition~\ref{th:wsfm-in-fm}), then for each finite model $M'$ of $D\cup\Sigma$, we can find a well-supported finite model, that is minimal, $M$ of $D\cup\Sigma$, such that $M\subseteq M'$, and, in particular, there exists a homomorphism $h$ (i.e., the identity homomorphism) such that $h(M)\subseteq M'$. 
\end{proof}


\beforeSubsec

\subsection{Propagation ordering}\label{sec:propOrd}

Let us start with the preliminary notions of 
existentially well-supported atom and
propagated term.
Let $I$ be a well-supported finite instance,
and $(\alpha_1,\ldots,\alpha_m)$ be a well-supported ordering of $I$. 
An atom $\alpha$ of $I \setminus D$ is said \textit{existentially well-supported} 
w.r.t. the ordering $(\alpha_1,\ldots,\alpha_m)$ if, 
for each well-supporting rule $\rho$ for $\alpha$ w.r.t. $(\alpha_1,\ldots,\alpha_m)$,
it holds that ${\bf EV}(\rho)\neq\emptyset$.
Moreover, 
let $\alpha_j[k]=t$, for some position $k$,
then $t$ is said \textit{propagated} from an atom $\alpha_i$ in position $l$, 
whenever 
$i<j$, 
$\alpha_i[l]=t$, and 
there exist 
a well-supporting rule $\rho$ for $\alpha_j$ 
and a homomorphism $h$ such that 
$\alpha_i \in h(body(\rho))$.
%
Consider again ontology $\Sigma$ of Section~\ref{sec:overview}, and the well-supported finite model $M$ considered after Proposition~\ref{th:wsfm-in-fm}. For instance, the atom $f(c_1,c_1)$ is existentially well-supported.
Indeed, the unique way to well-support the atom comes from the first rule of $\Sigma$, that is an existential rule.
%
We are now ready to define the notion of propagation ordering.

\begin{definition}[Propagation ordering]
Let $D$ be a database, $\Sigma$ be a $\mathsf{joinless}$ ontology, $M\in \mathit{wsfmods}(D,\Sigma)$, and $(\alpha_1,\ldots,\alpha_m)$ be a well-supported ordering of $M$. 
For each $\alpha_j\in M$, we build a new atom $\langle \alpha_j \rangle$ as follows.
Let $t=\alpha_j[k]$. We have:
$(1)$ If $\alpha_j$ is an existentially well-supported atom and $k$ is an existential position, 
then $\langle \alpha_j \rangle[k] = \langle t,j,k  \rangle$,
where $\langle t,j,k  \rangle$ is called a \textit{starting point} of $t$;
$(2)$ If $t$ is a propagated term from some atom $\alpha_i$ in position $l$, 
then $\langle \alpha_j \rangle [k] = \langle \alpha_i \rangle [l]$; and
$(3)$ $\langle \alpha_j \rangle [k] = \alpha_j[k]$, otherwise.
We call $(\langle \alpha_1 \rangle,\ldots,\langle \alpha_m \rangle)$ a \textit{propagation ordering} of the well-supported ordering 
$(\alpha_1,\ldots,\alpha_m)$.
\end{definition}

Note that the same term could have several starting points.
This propagation ordering will be useful to remember a starting point of that particular term and its propagations in other atoms.

\begin{example}
Consider the following $\mathsf{joinless}$ ontology 
$\Sigma=\{s(X_1) \rightarrow \exists Y_1 p(X_1,Y_1)$;
$s(X_2) \rightarrow \exists Y_2 u(Y_2,X_2)$;
$p(X_3,Y_3),u(W_3,Z_3)\rightarrow r(Y_3,Z_3)$;
$p(X_4,Y_4)\rightarrow t(Y_4)\}$,
and the database $D=\{s(c_1)\}$.
As example, 
$M=\{s(c_1),t(c_2),t(n_1),$
$p(c_1,c_2),$
$p(c_1,n_1),$
$r(c_2,c_1),$ $r(n_1,c_1),$ $u(c_2,c_1),$ $u(n_1,$ $c_1)\}$ 
is a well-supported finite model of $D\cup\Sigma$.
Indeed, for instance, $(s(c_1),$ $p(c_1,c_2),$ $p(c_1,n_1),$ $u(c_2,c_1),$ $t(c_2),$ $u(n_1,c_1),$ $r(n_1,c_1),$ $t(n_1),r(c_2,c_1))$ is a well-supported ordering of $M$.
The existentially well-supported atoms are $p(c_1,c_2)$, $p(c_1,n_1)$, $u(c_2,c_1)$ and $u(n_1,c_1)$.
More specifically, 
$p(c_1,c_2)$ has the term $c_2$ in the existential position $2$, then
$\langle p(c_1,c_2) \rangle = p(c_1,\langle c_2,2,2 \rangle)$, as $p(c_1,c_2)$ is the second atom of the well-supported ordering considered;
$p(c_1,n_1)$ has the term $n_1$ in the existential position $2$, then
$\langle p(c_1,n_1) \rangle = p(c_1,\langle n_1,3,2 \rangle)$;
$u(c_2,c_1)$ has the term $c_2$ in the existential position $1$, then
$\langle u(c_2,c_1) \rangle = u(\langle c_2,4,1 \rangle,c_1)$;
$u(n_1,c_1)$ has the term $n_1$ in the existential position $1$, then
$\langle u(n_1,c_1) \rangle = u(\langle n_1,6,1 \rangle,c_1)$.
On the other hand, the term $c_2$ is propagated in the atom $t(c_2)$ in the first (and unique) position.
It comes from atom $p(c_1,c_2)$, and we know that the starting point of $c_2$ is $\langle c_2,2,2 \rangle$.
Therefore, $\langle t(c_2) \rangle = t(\langle c_2,2,2 \rangle)$.
Moreover, in a similarly way, we obtain that $\langle t(n_1) \rangle = t(\langle n_1,3,2 \rangle)$.
Finally, the term $n_1$ is propagated in the atom $r(n_1,c_1)$ in the first position, and it comes from atom $p(c_1,n_1)$; whereas
the term $c_1$ is propagated in the atom $r(n_1,c_1)$ in the second position, and it comes from atom $u(n_1,c_1)$. Therefore, $\langle r(n_1,c_1) \rangle = r(\langle n_1,3,2 \rangle, c_1)$.  
\hfill $\lhd$
\end{example}

With our technical tools in place, we are now able to prove 
the following technical result.

\begin{theorem}\label{th:f-to-fsn}
	For each $\Sigma \in \shy$, 
	if $D^c\cup\Sigma^c \models_\mathsf{wsf} q^c$
	then $D^c\cup\Sigma^c_a\models_\mathsf{wsf} q^c$.
\end{theorem}

\beforeProof

\begin{proof}[Proof intuition]
%
Consider an arbitrary model $M\in \mathit{wsfmods}(D^c,\Sigma^c_a)$.
%
%
It suffices to prove that there exist $M'\in \mathit{wsfmods}(D^c,\Sigma^c)$ and 
a homomorphism $h'$ s.t. $h'(M') \subseteq M$. 
Indeed, by hypothesis, there exists a homomorphism $h$ s.t. $h(q)\subseteq M'$, and so
$(h'\circ h)(q)\subseteq M$. 
%
 
The difficulty here is that $M$ could not be a model of $D^c \cup \Sigma^c$.
Consider the database $D = \{s(c)\}$ and the $\shy$ ontology 
$\Sigma = \{s(X) \rightarrow \exists Y  p(Y)$;  
$s(X) \rightarrow \exists Y  r(Y)$; 
$p(X), r(X) \rightarrow g(X) \}$.
The canonical rewriting is $D^c = \{s_{[c]}\}$ and 
$\Sigma^c$ as follows:

\befAftVectors

{\small
\[\begin{array}{rclrclrclr}
s_{[c]} & \rightarrow & \exists Y  p_{[1]}(Y) & \ \ \ \ \ \ \ \ \ \
s_{[c]} & \rightarrow & \exists Y  r_{[1]}(Y) & \ \ \ \ \ \ \ \ \ \ \ \ \
p_{[c]}, r_{[c]} & \rightarrow & g_{[c]} & \\
s_{[1]}(X) & \rightarrow & \exists Y  p_{[1]}(Y) & \ \ \ \ \ \ \ \ \ \
s_{[1]}(X) & \rightarrow & \exists Y  r_{[1]}(Y) & \ \ \ \ \ \ \ \ \ \ \ \
p_{[1]}(X), r_{[1]}(X) & \rightarrow & g_{[1]}(X) &
\end{array}
\]}

\befAftVectors

\noindent One can verify that $M = \{s_{[c]}, p_{[1]}(n_1), r_{[1]}(n_1)\}$ is a 
(minimal) well-supported finite model of $D^c \cup \Sigma^c_a$
since, by Proposition~\ref{th:can-rew-of-shy-is-shy}, $\Sigma^c$ is $\shy$, and
since $\Sigma^c_a$ is obtained from $\Sigma^c$ by discarding the last harmless rule.
However, $M$ is not
a model of $D^c \cup \Sigma^c$ because the last rule is not satisfied.

The idea is to show how to construct from $M$
a model $M'\in \mathit{wsfmods}(D^c,\Sigma^c)$ that can be homomorphically 
mapped to $M$.
Intuitively, we identify the {\em starting points} in which 
existential variables of $\Sigma^c_a$ have been satisfied
and rename the introduced terms using a propagation ordering.

In the example above, consider the
well-supported ordering $(s_{[c]}, p_{[1]}(n_1), r_{[1]}(n_1))$ of $M$,
replace $n_1$ in $p_{[1]}(n_1)$ by $\langle n_1,2,2\rangle$ 
(null $n_1$ introduced in the second atom in the second position), and
replace $n_1$ in $r_{[1]}(n_1)$ by $\langle n_1,3,2\rangle$ 
(null $n_1$ introduced in the third atom in the second position).
Then, since $M$ is well-supported, we propagate (if needed)
these new terms according the supporting ordering.
In our case, $M' =  \{s_{[c]}, p_{[1]}(\langle n_1,2,2\rangle), r_{[1]}(\langle n_1,3,2\rangle)\}$
is now a finite model of $D^c \cup \Sigma^c$ that can be mapped to $M$.
\end{proof}

\beforeSubsec

\subsection{The main result}



\begin{lemma}\label{lm:main}
		Under $\mathsf{shy}$ ontologies, $D^c \cup\Sigma^c \models q^c$ if, and only if,
		$D^c\cup\Sigma^c\fmodels q^c$.
\end{lemma}

\beforeProof

\begin{proof}
Clearly, the ``only if'' implication is straightforward.
Hence, given a $\shy$ ontology $\Sigma$, we have to prove that 
$D^c \cup\Sigma^c \models q^c$,
whenever $D^c \cup\Sigma^c \models_\mathsf{fin} q^c$, 
for each database $D$ and UBCQ $q$.
Suppose that 
$D^c \cup\Sigma^c \models_\mathsf{fin} q^c$, i.e., the query $q^c$ is satisfied by each finite model of $D^c\cup\Sigma^c$.
Thus, by Theorem~\ref{th:fsn-to-f},
holds that
$D^c\cup\Sigma^c\models_\mathsf{wsf} q^c$, that is, the canonical rewriting of the query $q$ is satisfied by each well-supported finite model of the logical theory $D^c\cup\Sigma^c$.
Then, by Theorem~\ref{th:f-to-fsn},
holds that
$D^c\cup\Sigma^c_a\models_\mathsf{wsf} q^c$, that is, the canonical rewriting of the query $q$ is satisfied by each well-supported finite model of the $\mathsf{joinless}$ logical theory $D^c\cup\Sigma^c_a$.
Moreover, again, by Theorem~\ref{th:fsn-to-f},
we obtain that
$D^c\cup\Sigma^c_a\models_\mathsf{fin} q^c$, that is $q^c$ is satisfied also by every finite model of the previous theory.
Now, as $\Sigma^c_a$ is a $\mathsf{joinless}$ ontology, by the finite controllability of $\mathsf{joinless}$ ontologies proved by \citeN{DBLP:conf/lics/GogaczM13},
holds that 
$D^c\cup\Sigma^c_a\models q^c$.
Finally, by Theorem~\ref{thm:activeCanonicalGeneral},
we have that
$D^c\cup\Sigma^c\models q^c$, i.e. the query $q^c$ is satisfied by each model (finite or infinite) of $D^c\cup\Sigma^c$.
\end{proof}

Summing-up, Theorem~\ref{th:main} follows by combining 
Lemma~\ref{lm:main} with the properties of the canonical rewriting 
proved in Section~\ref{sec:tools}.

\section{Related work}\label{sec:relWork}

To complete the related works started with the Introduction,
we recall that the notion of finite controllability was formalized for the first time
by \citeN{Rosati:2006:DFC:1142351.1142404}
while he was working on a question that had been left open 
two decades before by \citeN{DBLP:journals/jcss/JohnsonK84}
about containment of conjunctive queries 
in case of both arbitrary and finite databases.
Basically, using our terminology, they proved that
ontologies mixing both inclusion-dependencies and functional-dependencies
are not finitely controllable,
by leaving open the case where ontologies contain inclusion-dependencies only.
Rosati then answered positively this question.

The semantic equivalence of fundamental reasoning tasks 
under finite and infinite models is not at all a prerogative
of the database community.
A sister yet orthogonal property of finite controllability
is of paramount importance also in logic, where it has been investigated
much earlier.
It is known as {\em finite model property} or {\em finite satisfiability}~\cite{Ebbinghaus1995}, and 
it asks for a class $\mathcal{C}$ of sentences whether every satisfiable sentence of $\mathcal{C}$ has a finite model.
For example, both G\"{o}del and Sch\"{u}tte proved that $\forall^2 \exists^*$ first-order
sentences are finitely satisfiable.

\new{Although reasoning under finite models has a long history
and it has been actively investigated in various fields of Computer Science, 
finite controllability remains 
open for many languages combining or generalizing
the key properties underlying the 
basic classes depicted in Figure~\ref{fig:taxonomy}.
By way of example, we mention
$(i)$ $\mathsf{glut}$-$\mathsf{guarded}$~\cite{DBLP:conf/ijcai/KrotzschR11}, extending $\mathsf{guarded}$ and $\mathsf{weakly}$-$\mathsf{acyclic}$;
$(ii)$ $\mathsf{weakly}$-$\mathsf{sticky}$-$\mathsf{join}$~\cite{DBLP:journals/ai/CaliGP12}, extending $\mathsf{sticky}$-$\mathsf{join}$, $\mathsf{weakly}$-$\mathsf{acyclic}$ and $\mathsf{shy}$; and 
$(iii)$ $\mathsf{tame}$~\cite{DBLP:journals/tplp/GottlobMP13}, extending $\mathsf{sticky}$ and $\mathsf{guarded}$.

Between $\mathsf{guarded}$ and 
$\mathsf{glut}$-$\mathsf{guarded}$, it is worth to recall $\mathsf{weakly}$-$\mathsf{guarded}$~\cite{DBLP:journals/jair/CaliGK13}, 
where each rule body has an atom covering all those variables 
that only occur in invaded (a.k.a. affected) positions.
Actually, this class is finitely controllable
although the proof sketch given by~\citeN{DBLP:journals/corr/BaranyGO13}
has some hole (there, some model of $D \cup \Sigma'$ might not satisfy $\Sigma$).
In fact, our canonical rewriting 
yields an ontology that can be partitioned in 
{\em active} and {\em harmless}, where 
the active part is $\mathsf{guarded}$.
Well-supported models and propagation ordering behave as for $\shy$.

An additional clarification concerns the notions of linear and sticky-join
considered by \citeN{JCSS17}, since they are not standard (actually stricter).
In the former, repeated variables are admitted only in rule heads,
while for the latter the authors state that the difference between 
sticky and sticky-join ``can only be seen if repeated variables 
in the heads of the rules are allowed''. 
(Regarding $\mathsf{sticky}$, the classical notion 
is only rephrased: their ``immortal'' positions 
correspond to positions being not marked.)
From such a mismatch, however, it follows that finite controllability 
of $\mathsf{sticky}$-$\mathsf{join}$ was unknown before our work.
A curious reader may verify that the proof
of their Lemma 4  breaks down when moving to a
$\mathsf{linear}$ (hence $\mathsf{sticky}$-$\mathsf{join}$) ontology
such as $\Sigma =\{ p(X,X) \rightarrow r(X)$; $r(X)\rightarrow \exists Y \, r(Y) \}$
---inducing no immortal position since all positions $p[1]$, $p[2]$ and $r[1]$ host marked variables--- paired with the singleton database $D = \{p(c,c)\}$.}

\section{Conclusion}\label{sec:DiscConc}

\new{By demonstrating that $\shy$ is finitely controllable, 
we complete an important picture around the basic 
decidable Datalog$^\pm$ classes.
But we take it as a starting point rather than an ending one.

On the one hand, finite controllability 
immediately implies decidability of OBQA.
Actually, via the soundness and completeness of the chase procedure we know that
the problem of deciding whether a UBCQ is {\em true} over a Datalog$^\pm$ theory 
is recursively enumerable.
But the complementary
problem of deciding whether a UBCQ is {\em false}
over a finitely controllable Datalog$^\pm$ class $\mathcal{C}$
is recursively enumerable too. 
In fact, each theory $D \cup \Sigma$, with $\Sigma \in \mathcal{C}$,
always admits a fair lexicographic enumeration of its finite models.
Unfortunately, such a na\"{\i}ve procedure would be 
inefficient in practice. 
Making it usable and competitive for real world 
problems is challenging and it is part of our ongoing work. 
Basically, this would lead to a tool able to deal with
any finitely controllable fragment, some of which (e.g., $\mathsf{guarded}$) 
have no effective implementation.

On the other hand, we believe the techniques developed in this paper
could have future applications.
For example, we are working on an extended version 
of our canonical rewriting that encodes in the predicates also a limited amount of 
nulls. This requires more complex techniques, which however would apply 
to classes using the key properties underlying 
$\mathsf{weakly}$-$\mathsf{acyclic}$, such as 
$\mathsf{glut}$-$\mathsf{guarded}$ and $\mathsf{weakly}$-$\mathsf{sticky}$-$\mathsf{join}$ (see Section~\ref{sec:relWork}).
Hence, by combining these techniques with the above tool for finitely controllable classes, we aim at the design and implementation of a reasoner 
able to deal with ontologies falling in any known decidable Datalog$^\pm$ class.}

Finally ---even if the unrestricted set of existential rules 
cannot be finitely controllable since it is not decidable---
it is still open, to the best of our knowledge,
whether there exists, or not, a 
fragment of existential rules which is decidable but not finitely controllable.

\section*{Acknowledgement}
\thanks{
The paper has been partially supported by the Italian Ministry for Economic Development (MISE) under project ``PIUCultura -- Paradigmi Innovativi per l'Utilizzo della Cultura'' (n. F/020016/01-02/X27),
and under project ``Smarter Solutions in the Big Data World (S2BDW)'' (n. F/050389/01-03/X32) funded
within the call ``HORIZON2020'' PON I\&C 2014-2020.
}

\bibliographystyle{acmtrans}
\bibliography{bibtex}

\begin{thebibliography}{}

\bibitem[\protect\citeauthoryear{Baader, Calvanese, McGuinness, Nardi, and
  Patel-Schneider}{Baader et~al\mbox{.}}{2003}]{Baader:2003:DLH:885746}
{\sc Baader, F.}, {\sc Calvanese, D.}, {\sc McGuinness, D.~L.}, {\sc Nardi,
  D.}, {\sc and} {\sc Patel-Schneider, P.~F.}, Eds. 2003.
\newblock {\em The Description Logic Handbook: Theory, Implementation, and
  Applications}.
\newblock CUP.

\bibitem[\protect\citeauthoryear{Baget, Lecl{\`{e}}re, and Mugnier}{Baget
  et~al\mbox{.}}{2010}]{DBLP:conf/kr/BagetLM10}
{\sc Baget, J.}, {\sc Lecl{\`{e}}re, M.}, {\sc and} {\sc Mugnier, M.} 2010.
\newblock Walking the decidability line for rules with existential variables.
\newblock In {\em Proc. of KR'10}.

\bibitem[\protect\citeauthoryear{Baget, Lecl{\`{e}}re, Mugnier, and
  Salvat}{Baget et~al\mbox{.}}{2009}]{DBLP:conf/ijcai/BagetLMS09}
{\sc Baget, J.}, {\sc Lecl{\`{e}}re, M.}, {\sc Mugnier, M.}, {\sc and} {\sc
  Salvat, E.} 2009.
\newblock Extending decidable cases for rules with existential variables.
\newblock In {\em Proc. of IJCAI'09}. 677--682.

\bibitem[\protect\citeauthoryear{Baget, Lecl{\`{e}}re, Mugnier, and
  Salvat}{Baget et~al\mbox{.}}{2011}]{DBLP:journals/ai/BagetLMS11}
{\sc Baget, J.}, {\sc Lecl{\`{e}}re, M.}, {\sc Mugnier, M.}, {\sc and} {\sc
  Salvat, E.} 2011.
\newblock On rules with existential variables: Walking the decidability line.
\newblock {\em AIJ\/}~{\em 175,\/}~9-10, 1620--1654.

\bibitem[\protect\citeauthoryear{B{\'{a}}r{\'{a}}ny, Gottlob, and
  Otto}{B{\'{a}}r{\'{a}}ny et~al\mbox{.}}{2014}]{DBLP:journals/corr/BaranyGO13}
{\sc B{\'{a}}r{\'{a}}ny, V.}, {\sc Gottlob, G.}, {\sc and} {\sc Otto, M.} 2014.
\newblock Querying the guarded fragment.
\newblock {\em Logical Methods in Computer Science\/}~{\em 10,\/}~2.

\bibitem[\protect\citeauthoryear{Bienvenu, ten Cate, Lutz, and Wolter}{Bienvenu
  et~al\mbox{.}}{2014}]{DBLP:journals/tods/BienvenuCLW14}
{\sc Bienvenu, M.}, {\sc ten Cate, B.}, {\sc Lutz, C.}, {\sc and} {\sc Wolter,
  F.} 2014.
\newblock Ontology-based data access: {A} study through disjunctive datalog,
  csp, and {MMSNP}.
\newblock {\em {ACM} TODS\/}~{\em 39,\/}~4, 33:1--33:44.

\bibitem[\protect\citeauthoryear{Bourhis, Manna, Morak, and Pieris}{Bourhis
  et~al\mbox{.}}{2016}]{Bourhis:2016:GDT:3014437.2976736}
{\sc Bourhis, P.}, {\sc Manna, M.}, {\sc Morak, M.}, {\sc and} {\sc Pieris, A.}
  2016.
\newblock Guarded-based disjunctive tuple-generating dependencies.
\newblock {\em ACM TODS\/}~{\em 41,\/}~4, 27:1--27:45.

\bibitem[\protect\citeauthoryear{Cal{\`{\i}}, Gottlob, and Kifer}{Cal{\`{\i}}
  et~al\mbox{.}}{2013}]{DBLP:journals/jair/CaliGK13}
{\sc Cal{\`{\i}}, A.}, {\sc Gottlob, G.}, {\sc and} {\sc Kifer, M.} 2013.
\newblock Taming the infinite chase: Query answering under expressive
  relational constraints.
\newblock {\em J. Artif. Intell. Res. {(JAIR)}\/}~{\em 48}, 115--174.

\bibitem[\protect\citeauthoryear{Cal{\`{\i}}, Gottlob, and
  Lukasiewicz}{Cal{\`{\i}} et~al\mbox{.}}{2009a}]{DBLP:conf/icdt/CaliGL09}
{\sc Cal{\`{\i}}, A.}, {\sc Gottlob, G.}, {\sc and} {\sc Lukasiewicz, T.}
  2009a.
\newblock Datalog\({}^{\mbox{{\(\pm\)}}}\): a unified approach to ontologies
  and integrity constraints.
\newblock In {\em Proc. of ICDT'09}. 14--30.

\bibitem[\protect\citeauthoryear{Cal{\`{\i}}, Gottlob, and
  Lukasiewicz}{Cal{\`{\i}} et~al\mbox{.}}{2009b}]{DBLP:conf/dlog/CaliGL09}
{\sc Cal{\`{\i}}, A.}, {\sc Gottlob, G.}, {\sc and} {\sc Lukasiewicz, T.}
  2009b.
\newblock Tractable query answering over ontologies with datalog+/-.
\newblock In {\em Proc. of DL'09}.

\bibitem[\protect\citeauthoryear{Cal{\`{\i}}, Gottlob, and
  Lukasiewicz}{Cal{\`{\i}} et~al\mbox{.}}{2012}]{DBLP:journals/ws/CaliGL12}
{\sc Cal{\`{\i}}, A.}, {\sc Gottlob, G.}, {\sc and} {\sc Lukasiewicz, T.} 2012.
\newblock A general datalog-based framework for tractable query answering over
  ontologies.
\newblock {\em J. Web Sem.\/}~{\em 14}, 57--83.

\bibitem[\protect\citeauthoryear{Cal{\`{\i}}, Gottlob, and Pieris}{Cal{\`{\i}}
  et~al\mbox{.}}{2010}]{DBLP:journals/pvldb/CaliGP10}
{\sc Cal{\`{\i}}, A.}, {\sc Gottlob, G.}, {\sc and} {\sc Pieris, A.} 2010.
\newblock Advanced processing for ontological queries.
\newblock {\em {PVLDB}\/}~{\em 3,\/}~1, 554--565.

\bibitem[\protect\citeauthoryear{Cal{\`{\i}}, Gottlob, and Pieris}{Cal{\`{\i}}
  et~al\mbox{.}}{2012}]{DBLP:journals/ai/CaliGP12}
{\sc Cal{\`{\i}}, A.}, {\sc Gottlob, G.}, {\sc and} {\sc Pieris, A.} 2012.
\newblock Towards more expressive ontology languages: The query answering
  problem.
\newblock {\em AIJ\/}~{\em 193}, 87--128.

\bibitem[\protect\citeauthoryear{Calvanese, {De Giacomo}, Lembo, Lenzerini, and
  Rosati}{Calvanese et~al\mbox{.}}{2013}]{DBLP:journals/ai/CalvaneseGLLR13}
{\sc Calvanese, D.}, {\sc {De Giacomo}, G.}, {\sc Lembo, D.}, {\sc Lenzerini,
  M.}, {\sc and} {\sc Rosati, R.} 2013.
\newblock Data complexity of query answering in description logics.
\newblock {\em AIJ\/}~{\em 195}, 335--360.

\bibitem[\protect\citeauthoryear{Civili and Rosati}{Civili and
  Rosati}{2012}]{DBLP:conf/datalog/CiviliR12}
{\sc Civili, C.} {\sc and} {\sc Rosati, R.} 2012.
\newblock A broad class of first-order rewritable tuple-generating
  dependencies.
\newblock In {\em Proc. of Datalog 2.0}. 68--80.

\bibitem[\protect\citeauthoryear{Deutsch, Nash, and Remmel}{Deutsch
  et~al\mbox{.}}{2008}]{DBLP:conf/pods/DeutschNR08}
{\sc Deutsch, A.}, {\sc Nash, A.}, {\sc and} {\sc Remmel, J.~B.} 2008.
\newblock The chase revisited.
\newblock In {\em Proc. of PODS'08}. 149--158.

\bibitem[\protect\citeauthoryear{Ebbinghaus and Flum}{Ebbinghaus and
  Flum}{1995}]{Ebbinghaus1995}
{\sc Ebbinghaus, H.-D.} {\sc and} {\sc Flum, J.} 1995.
\newblock {\em Satisfiability in the Finite}.
\newblock Springer Berlin Heidelberg, 95--103.

\bibitem[\protect\citeauthoryear{Fages}{Fages}{1991}]{DBLP:journals/ngc/Fages91}
{\sc Fages, F.} 1991.
\newblock A new fixpoint semantics for general logic programs compared with the
  well-founded and the stable model semantics.
\newblock {\em New Generation Comput.\/}~{\em 9,\/}~3/4, 425--444.

\bibitem[\protect\citeauthoryear{Fagin, Kolaitis, Miller, and Popa}{Fagin
  et~al\mbox{.}}{2005}]{DBLP:journals/tcs/FaginKMP05}
{\sc Fagin, R.}, {\sc Kolaitis, P.~G.}, {\sc Miller, R.~J.}, {\sc and} {\sc
  Popa, L.} 2005.
\newblock Data exchange: semantics and query answering.
\newblock {\em TCS\/}~{\em 336,\/}~1, 89--124.

\bibitem[\protect\citeauthoryear{Gogacz and Marcinkowski}{Gogacz and
  Marcinkowski}{2013}]{DBLP:conf/lics/GogaczM13}
{\sc Gogacz, T.} {\sc and} {\sc Marcinkowski, J.} 2013.
\newblock Converging to the chase - {A} tool for finite controllability.
\newblock In {\em Proc. of {LICS}'13}. 540--549.

\bibitem[\protect\citeauthoryear{Gogacz and Marcinkowski}{Gogacz and
  Marcinkowski}{2017}]{JCSS17}
{\sc Gogacz, T.} {\sc and} {\sc Marcinkowski, J.} 2017.
\newblock Converging to the chase - {A} tool for finite controllability.
\newblock {\em JCSS\/}~{\em 83,\/}~1, 180--206.

\bibitem[\protect\citeauthoryear{Gottlob, Kikot, Kontchakov, Podolskii,
  Schwentick, and Zakharyaschev}{Gottlob
  et~al\mbox{.}}{2014}]{DBLP:journals/ai/GottlobKKPSZ14}
{\sc Gottlob, G.}, {\sc Kikot, S.}, {\sc Kontchakov, R.}, {\sc Podolskii,
  V.~V.}, {\sc Schwentick, T.}, {\sc and} {\sc Zakharyaschev, M.} 2014.
\newblock The price of query rewriting in ontology-based data access.
\newblock {\em AIJ\/}~{\em 213}, 42--59.

\bibitem[\protect\citeauthoryear{Gottlob, Manna, and Pieris}{Gottlob
  et~al\mbox{.}}{2013}]{DBLP:journals/tplp/GottlobMP13}
{\sc Gottlob, G.}, {\sc Manna, M.}, {\sc and} {\sc Pieris, A.} 2013.
\newblock Combining decidability paradigms for existential rules.
\newblock {\em {TPLP}\/}~{\em 13,\/}~4-5, 877--892.

\bibitem[\protect\citeauthoryear{Gottlob, Orsi, and Pieris}{Gottlob
  et~al\mbox{.}}{2014}]{DBLP:journals/tods/GottlobOP14}
{\sc Gottlob, G.}, {\sc Orsi, G.}, {\sc and} {\sc Pieris, A.} 2014.
\newblock Query rewriting and optimization for ontological databases.
\newblock {\em {ACM} TODS\/}~{\em 39,\/}~3, 25:1--25:46.

\bibitem[\protect\citeauthoryear{Gottlob, Pieris, and Tendera}{Gottlob
  et~al\mbox{.}}{2013}]{DBLP:conf/icalp/GottlobPT13ICALP}
{\sc Gottlob, G.}, {\sc Pieris, A.}, {\sc and} {\sc Tendera, L.} 2013.
\newblock Querying the guarded fragment with transitivity.
\newblock In {\em Proc. of ICALP'13}. 287--298.

\bibitem[\protect\citeauthoryear{Johnson and Klug}{Johnson and
  Klug}{1984}]{DBLP:journals/jcss/JohnsonK84}
{\sc Johnson, D.~S.} {\sc and} {\sc Klug, A.~C.} 1984.
\newblock Testing containment of conjunctive queries under functional and
  inclusion dependencies.
\newblock {\em JCSS\/}~{\em 28,\/}~1, 167--189.

\bibitem[\protect\citeauthoryear{Kr{\"{o}}tzsch and Rudolph}{Kr{\"{o}}tzsch and
  Rudolph}{2011}]{DBLP:conf/ijcai/KrotzschR11}
{\sc Kr{\"{o}}tzsch, M.} {\sc and} {\sc Rudolph, S.} 2011.
\newblock Extending decidable existential rules by joining acyclicity and
  guardedness.
\newblock In {\em Proc. of IJCAI'11}. 963--968.

\bibitem[\protect\citeauthoryear{Leone, Manna, Terracina, and Veltri}{Leone
  et~al\mbox{.}}{2012}]{DBLP:conf/kr/LeoneMTV12}
{\sc Leone, N.}, {\sc Manna, M.}, {\sc Terracina, G.}, {\sc and} {\sc Veltri,
  P.} 2012.
\newblock Efficiently computable {D}atalog{\(^\exists\)} programs.
\newblock In {\em Proc. of KR'12}.

\bibitem[\protect\citeauthoryear{P{\'{e}}rez{-}Urbina, Motik, and
  Horrocks}{P{\'{e}}rez{-}Urbina et~al\mbox{.}}{2010}]{PerezUrbina2010186}
{\sc P{\'{e}}rez{-}Urbina, H.}, {\sc Motik, B.}, {\sc and} {\sc Horrocks, I.}
  2010.
\newblock Tractable query answering and rewriting under description logic
  constraints.
\newblock {\em J APPL LOGIC\/}.

\bibitem[\protect\citeauthoryear{Rosati}{Rosati}{2006}]{Rosati:2006:DFC:1142351.1142404}
{\sc Rosati, R.} 2006.
\newblock On the decidability and finite controllability of query processing in
  databases with incomplete information.
\newblock In {\em Proc. of PODS'06}.

\bibitem[\protect\citeauthoryear{Rosati}{Rosati}{2007}]{DBLP:conf/icdt/Rosati07}
{\sc Rosati, R.} 2007.
\newblock The limits of querying ontologies.
\newblock In {\em Proc. of ICDT'07}. 164--178.

\end{thebibliography}

\label{lastpage}

\newpage

\section*{Appendix}

\subsection*{{Shy} existential rules}\label{sec:Shy}
This section is devoted to recall the formal definition of $\mathsf{shy}$ ontologies and their syntactic properties,
as defined in~\citeN{DBLP:conf/kr/LeoneMTV12}.
%
For notational convenience and without loss of generality, we assume here that 
each pair of rules of an ontology share no variable.
%
Let $\Sigma$ be an ontology, 
$\alpha$ be a $m$-arity atom, 
$i\in\{1,\ldots,m\}$ be an index, 
$pred(\alpha)  = a$, and
$X$ be an existential variable occurring in some rule of $\Sigma$.
We say that position $a[i]$ is \textit{invaded} by $X$ if there exists a rule $\rho\in\Sigma$ such that $head(\rho)=\alpha$ and 
\begin{itemize}
\item[($i$)] $\alpha[i] = X$; or 
\item[($ii$)] $\alpha[i]$ is a universal variable of $\rho$ and 
all of its occurrences in $body(\rho)$ appear in positions invaded by $X$.
\end{itemize}
Let $\phi({\bf X})$ be a conjunction of atoms, and let $X\in {\bf X}$.
We say that $X$ is \textit{attacked} by a variable $Y$ in $\phi({\bf X})$
if all the positions where $X$ appears are invaded by $Y$.
On the other hand, we say that $X$ is \textit{protected} in $\phi({\bf X})$, if it is attacked by no variable.

A rule $\rho$ of an ontology $\Sigma$ is called \textit{shy} w.r.t. $\Sigma$ if the following conditions are both satisfied:
\begin{itemize}
\item[($i$)] 
if a variable $X$ occurs in more than one body atom,
then $X$ is protected in $body(\rho)$;
\item[($ii$)] 
if two distinct variables
are not protected in $body(\rho)$
but occur both in $head(\rho)$ and in two different body atoms,
then they are not attacked by the same variable.
\end{itemize}
Finally, if each $\rho \in\Sigma$ is shy w.r.t. $\Sigma$, then call $\Sigma$ a
$\shy$ ontology. 

\begin{example}
Consider the following rules
\begin{center}
$\begin{array}{lrcl}
\rho_1= & s(X_1) & \rightarrow & \exists Y_1 p(X_1,Y_1); \\
\rho_2= & p(X_2,Y_2),u(Y_2) & \rightarrow & r(X_2,Y_2); \\
\rho_3= & t(X_3) &\rightarrow &\exists Y_3 u(Y_3). \\
\end{array}$
\end{center}
Let $\Sigma=\{\rho_1,\rho_2,\rho_3\}$.
Clearly, $\rho_1$ and $\rho_3$ are $\shy$ rules w.r.t. $\Sigma$,
since they are also $\mathsf{linear}$ rules, namely rules with one single body atom,
which cannot violate any of the two shy conditions.
Moreover, rule $\rho_2$ is also $\shy$ w.r.t. $\Sigma$ as the positions $p[2]$ and $u[1]$ are invaded by disjoint sets of existential variables. Indeed, $p[2]$ is invaded by the existential variable $Y_1$ of the first rule, and $u[1]$ is invaded by the existential variable $Y_3$ of the third rule.
Therefore, $\Sigma$ is a $\shy$ ontology.

Now, consider the further three existential rules
\begin{center}
$\begin{array}{lrcl}
\rho_4= & u(X_4) & \rightarrow & \exists Y_4 p(Y_4,X_4); \\
\rho_5= & u(X_5) & \rightarrow & \exists Y_5 p(X_5,Y_5); \\
\rho_6= & r(X_6,X_6) & \rightarrow & v(X_6). \\
\end{array}$
\end{center}

Let $\Sigma'$ be the ontology $\Sigma\cup\{\rho_4\}$. 
It is easy to see that $\rho_1$, $\rho_3$ and $\rho_4$ are $\shy$ w.r.t. $\Sigma'$.
However, $\rho_2$ is not $\shy$ w.r.t. $\Sigma'$, as property $(i)$ is not satisfied.
Indeed, the variable $Y_2$ occurring in two body atoms in $body(\rho_2)$ is not protected, as the position $p[2]$ and $u[1]$ (the only positions in which $Y_2$ occurs) are invaded by the same existential variable, namely $Y_3$.
Therefore, $\Sigma'$ is not a $\shy$ ontology.

Let $\Sigma''$ be the ontology $\Sigma\cup\{\rho_5,\rho_6\}$. 
Again, $\rho_1$, $\rho_3$, $\rho_5$ and $\rho_6$ are trivially $\shy$ w.r.t. $\Sigma''$;
and again $\rho_2$ is not $\shy$ w.r.t. $\Sigma''$.
However, this time, $\rho_2$ is not $\shy$ because property $(ii)$ is not satisfied.
Indeed, the universal variables $X_2$ and $Y_2$, 
occurring in two different body atoms and in $head(\rho_2)$, 
are not protected in $body(\rho_2)$, 
as the position $p[1]$ and $u[1]$ (in which occur $X_2$ and $Y_2$, respectively) 
are attacked by the same variable $Y_3$.
Therefore, $\Sigma''$ is not a $\shy$ ontology.
%
\hfill $\lhd$
\end{example}

Essentially, during every possible chase step,
condition $(i)$ guarantees that each variable that occurs in more than one body atom
is always mapped into a constant.
Although this is the key property behind $\mathsf{shy}$, 
we now explain the role played by condition $(ii)$ and its importance.
To this aim, we exploit again $\Sigma''$, as introduced in the previous example, and
we reveal why this second condition, in a sense, turns into the first one.
Indeed, the rule $\rho_6$ bypasses the propagation of the same null in $\rho_2$ via different variables.
However, one can observe that the rules $\rho_2$ and $\rho_6$ imply the rule
$\rho_6': p(X_6,Y_6),u(X_6) \rightarrow v(X_6)$, which of course 
does not satisfy condition $(i)$.
Actually, it is not difficult to see that every ontology can be rewritten (independently from $D$ and $q$) into an en equivalent one (w.r.t. query answering) where
all the rules satisfy condition $(i)$. As an example, consider the following rule $\rho$
\begin{center}
	$\begin{array}{rcl}
	p(X_1,Y_1), r(Y_1,Z_1), u(Z_1, Y_1) & \rightarrow & \exists W_1 \, t(X_1,Z_1,W_1),
	\end{array}$
\end{center}
and assume that it belongs to some ontology $\Sigma$ and that it is not shy w.r.t. $\Sigma$
because it violates condition $(i)$ only.
Let us now construct $\Sigma'$ as $\Sigma \setminus \{\rho\}$ plus the following two rules:
\begin{center}
	$\begin{array}{rcl}
	p(X_1,Y_1), r(Y_1',Z_1), u(Z_1',Y_1'') & \hspace*{-0.15cm}\rightarrow & \hspace*{-0.15cm}aux_\rho (X_1,Y_1,Y_1',Z_1,Z_1',Y_1'');\\
	aux_\rho(X_1,Y_1,Y_1,Z_1,Z_1,Y_1) & \hspace*{-0.15cm}\rightarrow & \hspace*{-0.15cm}\exists W_1 \, t(X_1,Z_1,W_1).
	\end{array}$
\end{center}
Both the new rules satisfy now condition $(i)$ w.r.t. $\Sigma'$.
Moreover, it is not difficult to see that, for every database $D$ and for every UBCQ $q$, it holds
that $D \cup \Sigma \models q$ if and only if $D \cup \Sigma' \models q$. However, since
$\rho$ does not satisfy condition $(i)$, this immediately implies that
the first new rule does not satisfy condition $(ii)$.
%

The syntactic properties of $\mathsf{shy}$ make the class quite expressive
since it strictly contains both $\mathsf{linear}$ and $\mathsf{datalog}$.
Moreover, these properties are easy recognizable 
and guarantee efficient answering to conjunctive queries,
as experimentally shown in~\citeN{DBLP:conf/kr/LeoneMTV12}.
In fact, ontology-based query answering over $\mathsf{shy}$ ontologies
preservers the same data and combined complexity of OBQA over $\mathsf{datalog}$,
namely \textsc{PTime}-complete and \textsc{ExpTime}-complete, respectively.

\subsection*{Formal Proofs}

\begin{proof}[Proof of Proposition~\ref{th:inf-to-inf}]
We prove that
$\rp(\mathit{chase}(D^c,\Sigma^c)) = \mathit{chase}(D,\Sigma)$
by induction on the chase step.
Let $I_0=D\subset I_1\subset \ldots\subset I_m\subset \ldots$ be a chase procedure of $D$ and $\Sigma$;
and let $I^c_0=D^c\subset I^c_1\subset \ldots\subset I^c_m\subset \ldots$ be a chase procedure of $D^c$ and $\Sigma^c$.

Clearly, the base case follows, since, by definition of the canonical rewriting of $D$, $\rp(D^c)=D$.

Then, assume that $\rp(I^c_m)=I_m$. We have to prove that $\rp(I^c_{m+1})=I_{m+1}$.
By definition of chase step, there exist a rule $\rho\in\Sigma$ and a homomorphism $h$ from $body(\rho)$ to $I_m$, such that $\langle \rho,h\rangle(I_m)=I_{m+1}$. That is,
$I_{m+1}=I_m\cup\{h(head(\rho))\}$.
By construction of a canonical rule, 
there exists a safe substitution $\varsigma$ w.r.t. $\rho$, such that
$\varsigma(\rho)^c$ is a canonical rule and, by inductive hypothesis,
there exists a homomorphism $h^c$ from $body(\varsigma(\rho)^c)$ to $I^c_m$.
Consider the following homomorphism
$(h^c)'= (h\setminus h|_{\bf X})\cup 
h^c|_{\bf X} \supseteq h^c|_{\bf X}$.
Therefore,
$I^c_{m+1}=I^c_m\cup \{(h^c)'(head(\langle \rho,\varsigma \rangle))\}$.
Moreover,
\begin{center}
$\begin{array}{lll}
\rp(I^c_{m+1}) & =  \rp(I^c_m\cup \{(h^c)'(head(\varsigma(\rho)^c))\}) & = \\
&=  \rp(I^c_m)\cup \rp(\{(h^c)'(head(\varsigma(\rho)^c))\}) & = \\
&=  I_m \cup \{h'(\rp(head(\varsigma(\rho)^c)))\} & = \\
& = I_m \cup \{h'(head(\rho))\} & =I_{m+1}.
\end{array}$
\end{center}
Finally, let $q^c$ be the canonical rewriting of the UBCQ $q=\exists {\bf Y}_1 \psi_1({\bf Y}_1) \vee \ldots \vee \exists {\bf Y}_k \psi_k({\bf Y}_k)$. For each $j\in\{1,\ldots,k\}$, consider the safe substitution $\varsigma_j$ mapping each variable of $\psi_j({\bf Y}_j)$ in a different null. 
Therefore, there exists a conjunction of atoms, say $\psi_j^c({\bf Y}_j)= \varsigma_j(\psi_j({\bf Y}_j))^c$ in $q^c$, such that 
$\rp(\psi_j^c({\bf Y}_j))=\psi_j({\bf Y}_j)$, for each $j\in\{1,\ldots,k\}$.
Hence, $q \subseteq \rp(q^c)$. 
Moreover, it is easy to see that, each other safe substitution $\varsigma'$ w.r.t. some $\psi_j$, produces a conjunction of atoms, $\varsigma'(\psi_j({\bf Y}_j))^c$ such that
$\rp(\varsigma'(\psi_j({\bf Y}_j))^c)$ is contained in 
$\rp(\varsigma_j(\psi_j({\bf Y}_j))^c)$. Therefore, $\rp(q^c) \subseteq q$.
Thus, $\rp(q^c)= q$.
\end{proof}

\new{
\begin{proof}[Proof of Theorem~\ref{thm:inf-to-inf}]
We know that, for each database $D$, ontology $\Sigma$ and UBCQ $q$, it holds that 
$D\cup\Sigma\models q$ if and only if
$\mathit{chase}(D,\Sigma)\models q$~\cite{DBLP:journals/tcs/FaginKMP05}.
Therefore, also	
$D^c\cup\Sigma^c\models q^c$ if and only if
$chase(D^c,\Sigma^c)\models q^c$.
Moreover, by Proposition~\ref{th:inf-to-inf}, we have that
$\rp(\mathit{chase}(D^c,\Sigma^c)) = \mathit{chase}(D,\Sigma)$ and $\rp(q^c)\equiv q$.
Hence, remain to prove that
$\rp(\mathit{chase}(D^c,\Sigma^c))\models \rp(q^c)$
if and only if
$chase(D^c,\Sigma^c)\models q^c$.

We prove the ``if'' part, given that the ``only if'' part can be obtained retracing the chain of the following implications.
Suppose that 
$chase(D^c,\Sigma^c)\models q^c$.
Therefore, there is a homomorphism $h$ from at least one disjunct of $q^c$, say $\varsigma_j(\psi_j(\Y_j))^c$ (where $\varsigma_j$ is a canonical substitution), to $\mathit{chase}(D^c,\Sigma^c)$,
that is $h(\varsigma_j(\psi_j(\Y_j))^c)$ $\subseteq$ $chase(D^c,\Sigma^c)$.
Therefore,  
$\rp(h(\varsigma_j(\psi_j(\Y_j))^c))$ $\subseteq$ 
$\rp(chase(D^c,\Sigma^c))$.
Moreover, note that 
$\rp(h(\varsigma_j(\psi_j(\Y_j))^c))$
$=$ $h(\rp(\varsigma_j(\psi_j(\Y_j))^c))$.
Hence,
$h(\rp(\varsigma_j(\psi_j(\Y_j))^c))$ $\subseteq$ 
$\rp(chase(D^c,\Sigma^c))$.
Thus, $h$ is also a homomorphism from
a disjunct of $\rp(q^c)$ to $\rp(chase(D^c,\Sigma^c))$,
that is $\rp(\mathit{chase}(D^c,\Sigma^c))\models \rp(q^c)$.
\end{proof}
}

\begin{proof}[Proof of Proposition~\ref{th:can-rew-of-shy-is-shy}]
Let $\Sigma$ be a $\mathsf{shy}$ ontology.
Note that, for each rule $\rho\in\Sigma $, there exists a rule 
$\varsigma(\rho)^c\in \Sigma^c$
such that $\varsigma(X^i)=n_i$ for each variable $X^i$ occurring in $\rho$.
It is easy to see that a such $\varsigma$ is a safe substitution.
We denote by $\bar{\Sigma}^c$ the set of all and anly this kind of rules in $\Sigma^c$.
Note that, if $\Sigma^c$ is a $\mathsf{shy}$ ontology, then
$\bar{\Sigma}^c\subseteq \Sigma^c$ is also a $\mathsf{shy}$ ontology.

By contradiction, suppose that $\bar{\Sigma}^c$ is not a $\mathsf{shy}$ ontology.

First, suppose that there exists a rule $\varsigma(\rho)^c\in\bar{\Sigma}^c$ such that
there exists a variable, say $X$, occurring in more than one body atom and
$X$ is not protected in $body(\varsigma(\rho)^c)$.
Therefore, for each existential variable $Y$,
there exists an atom $\beta\in body(\varsigma(\rho)^c)$ and
 some position $pred(\beta)[i]$ in which $X$ occurs, and $pred(\beta)[i]$ is not invaded by $Y$.
Consider the unpacked rule $\rp(\varsigma(\rho)^c)=\rho\in\Sigma$.
Therefore, by construction, 
for each existential variable $Y$,
there exists $\alpha\in body(\rho)$ and
some position $pred(\alpha)[j]$ in which $X$ occurs, and $pred(\alpha)[j]$ is not invaded by $Y$.
Hence, $X$ occurs in more than one body atom of $\rho$ and 
$X$ is not protected in $body(\rho)$.
So that, $\rho$ is not a $\mathsf{shy}$ rule, and, thus, $\Sigma$ is not a $\mathsf{shy}$ ontology.

Then, suppose that there exists a rule $\varsigma(\rho)^c\in\bar{\Sigma}^c$ such that
there are two distinct universal variables, say $X$ and $Y$, 
that are not protected in $body(\varsigma(\rho)^c)$; 
occur in $head(\varsigma(\rho)^c)$;
occur in two different body atoms; 
and they are attacked by the same variable.
Therefore, there exists an existential variable $Z$ such that $X$ and $Y$ occur only in invaded position by $Z$.
Consider again the unpacked rule $\rp(\varsigma(\rho)^c)=\rho\in\Sigma$.
Then, by the unpacking function, $X$ and $Y$ are not protected in $body(\rho)$, and 
they occur in $head(\rho)$, in two different body atoms, and 
only in invaded position by $Z$.
Thus, they are attacked by the same variable.
Therefore, also in this case, $\rho$ is not a $\mathsf{shy}$ rule.
Hence, $\Sigma$ is not a $\mathsf{shy}$ ontology.
\end{proof}

\begin{proof}[Proof of Proposition~\ref{th:wsfm-in-fm}]
Let $M$ be a finite model of $D\cup\Sigma$.
Clearly, if $M$ is a well-supported finite model of $D\cup\Sigma$, we are done.
Therefore, suppose that $M$ is not a well-supported finite model of $D\cup\Sigma$.
Let $\Omega_1=(\alpha_1,\ldots,\alpha_m)$ be an ordering of the atoms of $M$.
Hence, by assumption, 
there exists $\alpha\in M$ 
that is not a well-supported atom w.r.t. $\Omega_1$.
Let $\alpha_{j_1}$ be the first atom in the ordering $\Omega_1$ that is not well-supported.
And consider a new ordering
$\Omega_2=(\alpha_1,\ldots,\alpha_{j_1-1},\alpha_{j_1+1},\ldots,\alpha_m,\alpha_{j_1})$,
where $\alpha_{j_1}$ is shifted from the position ${j_1}$ to the  position $n$.
As $M\not\in \mathit{wsfmods}(D,\Sigma)$, then $\Omega_2$ is not a well-supported ordering of $M$.
Moreover, the first $j_1-1$ atoms are well-supported w.r.t. $\Omega_2$.
Therefore, let $\alpha_{j_2}$ be the first atom in the ordering $\Omega_2$ that is not well-supported. Again, we consider a new ordering, say $\Omega_3$, where $\alpha_{j_2}$ is shifted from position $j_2-1$ to the position $n$.
Iteratively, we build a sequence $\Omega_1,\Omega_2,\ldots,\Omega_m,\ldots$ of orderings that are not well-supported.
Note that, as the number of different orderings is finite, there exist at least two orderings in the sequence that are the same.
Therefore, let $\Omega_{m_1}$ and $\Omega_{m_2}$ be
the first two orderings of the sequence, with $m_2>m_1$, 
such that $\Omega_{m_1}=\Omega_{m_2}$ 
(i.e., $\Omega_{m_1}$ and $\Omega_{m_2}$ are the same ordering).
Consider the subset $A\subseteq M$ containing the 
first $n-(m_2-m_1)$ elements in $\Omega_{m_1}$, and
the set $B$ of the last $m_2-m_1$ atoms in $\Omega_{m_1}$.
By construction, $A$ is a well-supported instance.
Moreover, each $\beta\in B$ is not well-supported by $A$, as $\Omega_{m_2}=\Omega_{m_1}$.
That is, there is no rule $\rho$ in $\Sigma$ and no homomorphism $h$ such that 
$h(body(\rho))\subseteq A$ and $h(head(\rho))=\{\beta\}$.
Hence, as $M$ is a model,
whenever $A\models body(\rho)$, there exists an atom $\alpha$ in $A$,
such that $\alpha\models head(\rho)$.
Therefore, $A$ is a model.

To complete the proof, let $M$ be a finite minimal model of $D\cup\Sigma$.
As just proved, there exists a well-supported finite model $M'\subseteq M$.
By minimality of $M$, the model $M'$ must be equal to $M$.
Therefore, $M$ is a well-supported finite model.
\end{proof}

\begin{proof}[Proof of Theorem~\ref{th:f-to-fsn}]
We have to prove that for each $M\in \mathit{wsfmods}(D^c,\Sigma^c_a)$, 
there exist $M'\in \mathit{wsfmods}(D^c,\Sigma^c)$ and 
a homomorphism $h'$ such that $h'(M')\subseteq M$.
Indeed, by hypothesis, there exists a homomorphism $h$ such that $h(q)\subseteq M'$, and so
$(h'\circ h)(q)\subseteq M$. 

Let $M\in \mathit{wsfmods}(D^c,\Sigma^c_a)$, and
let $(\alpha_1,\ldots,\alpha_m)$ be a well-supported ordering of $M$, and
let $(\langle \alpha_1 \rangle,\ldots,\langle\alpha_m\rangle)$ be a propagation ordering of 
$(\alpha_1,\ldots,\alpha_m)$.
%
%
If there exists a join rule $\rho\in \Sigma^c$ satisfied by
$M$ with a null or a constant $t$ in the join variables, then
we consider the set of join atoms in the body of $\rho$ w.r.t. the term $t$, 
say $A\subseteq M$.
First, we substitute a term $t$ of some $\alpha\in A$ in position $l$, with the corresponding term $\langle t,j,k\rangle$ of $\langle\alpha\rangle$, that can be considered as a fresh null.
This new atom is denoted by $\alpha'$, so that $\alpha'[l]=\langle t,j,k\rangle$.
Then, for each $\alpha_i\in M$ such that $\langle \alpha_i \rangle[l] = \langle t,j,k\rangle$, for some position $l$,
we set $\alpha_i'[l] = \langle t,j,k\rangle $.
Otherwise, $\alpha_i'[l] = \alpha_i[l]$.
%
In this way, we build an instance $M'=\{\alpha':\alpha\in M\}$ of $\Sigma$, and a homomorphism $h'$ such that $h'(\langle t,j,k\rangle )=t$, 
for each introduced fresh null $\langle t,j,k\rangle $ to substitute $t$.
By construction, it holds that $h'(\alpha')=\alpha$, so that
$h'(M') = M$.
Note that, by construction, $M'$ is a well-supported finite instance of $D^c\cup\Sigma^c$.

Therefore, it remains to prove that $M'$ is a model of $D^c\cup\Sigma^c$.
By contradiction, suppose that $M'$ is not a model.
Hence, there exists a rule $\rho\in\Sigma^c$ such that $M'\models body(\rho)$, and $M'\not\models head(\rho)$. We distinguish two cases.
\begin{itemize}
\item[(i)] First, suppose that $\rho$ is not a join rule. 
Then, there exists a safe substitution $\hat{\varsigma}$, mapping each variable in the atoms of $\rho$ into a different null,
so that $\hat{\varsigma}(\rho)^c\in \Sigma^c_a$, as it is not a harmless rule of Shy.
By hypothesis, $M'\models body(\rho)$, so that there exists a homomorphism $h''$ such that $h''(body(\rho))\subseteq M'$.
Therefore, $h'(h''(body(\rho)))\subseteq h'(M')=M$, and so
$M\models body(\rho)$.
Hence, also $M\models body(\hat{\varsigma}(\rho)^c)$.
As $M$ is a model of $\Sigma^c_a$, then $M\models head(\hat{\varsigma}(\rho)^c)$.
Therefore, there exists a homomorphism $h'''$ such that $h'''(head(\hat{\varsigma}(\rho)^c))=\alpha_j$, for some $j\in\{1,\ldots,m\}$.
Hence, $\alpha_j\in M$. Therefore, $\alpha_j'\in M'$.
Moreover, $\alpha_j'\models head(\rho)$, as $h'(\alpha_j')=\alpha_j\models head(\rho)$.
Therefore, $M'\models head(\rho)$.

\item[(ii)] Now, suppose that $\rho$ is a join rule. 
Since, by hypothesis, $M'\models body(\rho)$, then, 
the join variables in the body of $\rho$ are instantiated by the same null, as $D^c\cup\Sigma^c$ is a constant-free logical theory.
However, by construction of $M$,
it is not possible that the same term comes from an instantiation of two different existential variables,
since we replaced each such instantiation with a fresh null in at least one joined term.

\end{itemize}
Therefore, $M'$ is a well-supported finite model of $D^c\cup\Sigma^c$.
\end{proof}

\end{document}